\documentclass{lmcs} 
\pdfoutput=1

\usepackage{lastpage}
\lmcsdoi{20}{3}{8}
\lmcsheading{}{\pageref{LastPage}}{}{}%
{Jul.~02,~2021}{Jul.~23,~2024}{}

\usepackage[utf8]{inputenc}

\keywords{Logic, Finite Model Theory, Game Comonads, Generalised Quantifiers }

\usepackage{amsmath, amsthm, amssymb, tikz, tikz-cd,pgf}
\usetikzlibrary{automata}
\usetikzlibrary{petri}
\usepackage{graphicx}
\usepackage[lofdepth,lotdepth]{subfig}
\usepackage{mathtools}
\usepackage{enumitem}
\usepackage{diagrams}
\usepackage{xspace}
\usepackage{xcolor}
\usepackage{ stmaryrd }
\usepackage{hyperref}
\newcommand{\defeq}{\vcentcolon=}

\newcommand{\abs}[1]{|#1|}

\newcommand{\T}[1]{\mathbb{P}_{#1}}
\newcommand{\G}[1]{\mathbb{H}_{#1}}
\newcommand{\good}{structured\xspace}

\newcommand{\tup}[1]{\mathbf{#1}}
\newcommand{\str}[1]{\mathcal{{#1}}}

\newcommand{\strs}[1]{\mathcal{R}(#1)}

\newcommand{\binj}[1]{\rightarrow^{\text{i}}_{#1}}
\newcommand{\bsurj}[1]{\rightarrow^{\text{s}}_{#1}}
\newcommand{\sbbij}[1]{\rightarrowtriangle^{\text{b}}_{#1}}
\newcommand{\sbinj}[1]{\rightarrowtriangle^{\text{i}}_{#1}}
\newcommand{\sbsurj}[1]{\rightarrowtriangle^{\text{s}}_{#1}}
\newcommand{\bfbij}[1]{\rightleftarrows^{\text{b}}_{#1}}
\newcommand{\bfinj}[1]{\rightleftarrows^{\text{i}}_{#1}}
\newcommand{\bfsurj}[1]{\rightleftarrows^{\text{s}}_{#1}}
\newcommand{\bbij}[1]{\rightarrow^{\text{b}}_{#1}}

\newcommand{\nats}{\mathbb{N}}
\newcommand{\game}[3]{\textbf{#1}_{#2}^{#3}}

\newcommand{\eplnkinftyomega}[1]{\mathcal{H}^{#1}}

\newcommand{\Datalog}{\texttt{Datalog}}
\newcommand{\posqfk}[1]{+\mathcal{L}^{#1}}
\newcommand{\qfk}[1]{\mathcal{L}^{#1}}
\newcommand{\FO}{\mathbf{FO}}
\newcommand{\epFO}{\exists^{+}\mathbf{FO}}
\newcommand{\infL}{\mathcal{L}_{\infty}}
\newcommand{\infH}{\exists^{+}\mathcal{L}_{\infty}}
\newcommand{\infLr}[1]{\mathcal{L}_{\infty, #1}}
\newcommand{\infHr}[1]{\exists^{+}\mathcal{L}_{\infty, #1}}
\newcommand{\infC}{\mathcal{C}}
\newcommand{\ra}{\rightarrow}
\newcommand{\pra}{\rightharpoonup}
\newcommand{\injpar}{x_{\sf{i}}}
\newcommand{\surjpar}{x_{\sf{s}}}
\newcommand{\negpar}{x_{\sf{n}}}

\theoremstyle{plain} 

\begin{document}

\title[Game Comonads \& Generalised Quantifiers]{Game Comonads \& Generalised Quantifiers}
\titlecomment{{\lsuper*}An earlier version of this paper, without proofs, appeared in the proceedings of CSL 2021.}

\author[A.~\'O Conghaile]{Adam \'O Conghaile\lmcsorcid{0000-0002-3032-5514}}	

\author[A.~Dawar]{Anuj Dawar\lmcsorcid{0000-0003-4014-8248}}	
\address{Department of Computer Science and Technology, University of Cambridge, United Kingdom }	
\email{anuj.dawar@cl.cam.ac.uk}  
\thanks{Research funded in part by EPSRC grant EP/T007257/1.}	

\begin{abstract}
  \noindent Game comonads, introduced by Abramsky, Dawar and Wang and developed by Abramsky and Shah, give an interesting categorical semantics to some Spoiler-Duplicator games that are common in finite model theory. In particular they expose connections between one-sided and two-sided games, and parameters such as treewidth and treedepth and corresponding notions of decomposition.  In the present paper, we expand the realm of game comonads to logics with generalised quantifiers.  In particular, we introduce a comonad graded by two parameters $n \leq k$ such that isomorphisms in the resulting Kleisli category are exactly Duplicator winning strategies in Hella's $n$-bijection game with $k$ pebbles.  We define a one-sided version of this game which allows us to provide a categorical semantics for a number of logics with generalised quantifiers.  We also give a novel notion of tree decomposition that emerges from the construction.
\end{abstract}

\maketitle

\section*{Introduction}\label{sec:intro}

Model-comparison games, such as Ehrenfeucht-Fra\"{\i}ss\'e games and pebble games play a central role in finite model theory.  Recent work by Abramsky et al.~\cite{Abramsky2017,Abramsky2018} provides a category-theoretic view of such games which yields new insights.  In particular, the \emph{pebbling comonad} $\T{k}$ introduced in~\cite{Abramsky2017} reveals an interesting relationship between one-sided and two-sided pebble games.  The morphisms in the Kleisli category associated with $\T{k}$ correspond exactly to winning strategies in the existential positive $k$-pebble game.  This game was introduced by Kolaitis and Vardi~\cite{KolaitisVardi1992} to study the expressive power of $\Datalog$.  A winning strategy for Duplicator in the game played on structures $\str{A}$ and $\str{B}$ implies that all formulas of existential positive $k$-variable logic true in $\str{A}$ are also true in $\str{B}$.  The game has found widespread application in the study of database query languages as well as constraint satisfaction problems.  Indeed, the widely used $k$-local consistency algorithms for solving constraint satisfaction can be understood as computing the approximation to homomorphism given by such strategies~\cite{Kolaitis2000}.  At the same time, isomorphisms in the Kleisli category associated with $\T{k}$ correspond to winning strategies in the $k$-pebble \emph{bijection} game. This game is a variant of the bijection game introduced by Hella~\cite{Hella1996} and characterises equivalence in the $k$-variable logic with counting.  This gives a family of equivalence relations (parameterised by  $k$) which has been widely studied as approximations of graph isomorphism.  It is often called the Weisfeiler-Leman family of equivalences and has a number of characterisations  in logic, algebra and combinatorics (see the discussion in~\cite{Grohe2017}).

The bijection game originally introduced by Hella is actually the initial level of a hierarchy of games that he defined to characterise equivalence in logics with generalised (i.e.\ Lindstr\"om) quantifiers.  For each $n, k \in \nats$ we have a $k$-pebble $n$-bijection game that characterises equivalence with respect to an infinitary $k$-variable logic with quantifiers of arity at most $n$.  In the present paper, we introduce a graded comonad associated with this game which we call the \emph{$n,k$-Hella comonad}, or $\G{n,k}$.  This comonad is obtained as a  quotient of the comonad $\T{k}$ and we are able to show that isomorphisms in the associated Kleisli category correspond to winning strategies for Duplicator in the $k$-pebble $n$-bijection game. The morphisms then correspond to a new one-way game we define, which we call the $k$-pebble $n$-function game.  We are able to show that this relates to a natural logic: a $k$-variable positive infinitary logic with $n$-ary \emph{homomorphism-closed quantifiers}.

This leads us to a systematic eight-way classification of model-comparison games based on what kinds of functions Duplicator is permitted (arbitrary functions, injections, surjections or bijections) and what the partial maps in game positions are required to preserve: just atomic information or also negated atoms.  We show that each of these variations correspond to preservation of formulas in a natural fragment of bounded-variable infinitary logic with $n$-ary Lindstr\"om quantifiers.  Moreover, winning strategies in these games also correspond to natural restrictions of the morphisms in the Kleisli category of $\G{n,k}$ that are well-motivated from the category-theoretic point of view.

Another key insight provided by the work of Abramsky et al.\ is that coalgebras in the pebbling comonad $\T{k}$ correspond exactly to tree decompositions of width $k$.   Similarly, the coalgebras in the Ehrenfeucht-Fra\"{\i}ss\'e comonad introduced by Abramsky and Shah characterise the \emph{treedepth} of structures.  This motivates us to look at coalgebras in $\G{n,k}$ and we show that they yield a new and natural notion of generalised tree decomposition.

In what follows, after a review of the necessary background in Section~\ref{sec:background}, we introduce the various games and logics in Section~\ref{sec:logic} and establish the relationships between them.  Section~\ref{sec:comonad} contains the definition of the Hella comonad and shows that interesting classes of morphisms in the associated Kleisli category correspond to winning strategies in the games.  The coalgebras of this comonad are investigated in Section~\ref{sec:coalgebra}, and the associated tree-decompositions of structures defined.

\section{Background}\label{sec:background}

In this section we introduce notation that we use throughout the paper and give a brief overview of background we assume.  For a positive integer $n$, we write $[n]$ for the set $\{1,\ldots,n\}$.

A tree $T$ is a set with a partial order $\leq$ such that for all $t \in T$, the set $\{x \mid x \leq t\}$ is linearly ordered by $\leq$ and such that there is an element $r \in T$ called \emph{the root} such that $r \leq t$ for all $t \in T$.  If $t <  t'$ in $T$ and there is no $x$ with $t < x< t'$, we call $t'$ a \emph{child} of $t$ and $t$ the \emph{parent} of $t'$.

For $X$ a set, we write $X^{\ast}$ for the set of \emph{lists} over elements of $X$ and $X^{+}$ for the set of non-empty lists. We write the list with elements $x_1, \ldots x_m$ in that order as $[x_1, \ldots x_m]$.  For two lists $s_1, s_2 \in X^{\ast}$ we write $s_1\cdot s_2$ for the list formed by concatenating $s_1$ and $s_2$. For $x \in X$ and $s \in X^{\ast}$ we write $x;s$ for the list with first element $x$ followed by the elements of $s$ in order and $s;x$ for the elements of the list $s$ in order followed by final element $x$. We occasionally underline the fact that $s_1\cdot s_2, x;s,$ and $ s;x$ are lists by writing them enclosed in square brackets, as $[s_1\cdot s_2], [x;s],$ and $ [s;x]$.

\subsection{Logics}
We work with finite relational signatures and assume a fixed signature $\sigma$.   Unless stated otherwise, the structures we consider are finite $\sigma$-structures.  We write $\str{A}, \str{B}, \str{C}$ etc.\ to denote such structures, and the corresponding roman letters  $A, B, C$ etc.\ to denote their universes.

We assume a standard syntax and semantics for first-order logic (as in~\cite{Libkin2004}), which we denote $\FO$.  We write $\infL$ for the infinitary logic that is obtained from $\FO$ by allowing conjunctions and disjunctions over arbitrary sets of formulas.   We write $\infH$ and $\epFO$ for the restriction of  $\infL$ and $\FO$ to existential positive formulas, i.e.\ those without negations or universal quantifiers.  We use natural number superscripts to denote restrictions of the logic to a fixed number of variables.  So, in particular $\FO^k, \infL^{k}$ and $\infH^{k}$ denote the $k$-variable fragments of $\FO$, $\infL$ and $\infH$ respectively.  Similarly, we use subscripts on the names of the logic to denote the fragments limited to a fixed nesting depth of quantifiers.  Thus, $\FO_r, \infLr{r}$ and $\infHr{r}$ denote the fragments  of $\FO$, $\infL$ and $\infH$ with quantifier depth at most $r$.  We write $\infC$ to denote the extension of $\infL$ where we are allowed quantifiers $\exists^{\geq i}$ for each natural number $i$.  The quantifier is to be read as ``there exists at least $i$ elements\ldots''.  We are mainly interested in the $k$-variable fragments of this logic $\infC^k$.

A formula $\phi(x_1,\ldots,x_n)$ with free variables among $x_1,\ldots,x_n$ in any of these logics defines an \emph{$n$-ary query}, that is a map from structures $\str{A}$ to $n$-ary relations on the structure.  A sentence defines a \emph{Boolean query}, i.e.\ a class of structures.  All queries are always closed under isomorphisms.

\subsection{Generalised quantifiers}
We use the term \emph{generalised quantifier} in the sense of Lindstr\"om~\cite{Lindstrom1966}.  These have been extensively studied in finite model theory (see~\cite{Hella1996,DH95,dawar1995}).  In what follows, we give a brief account of the basic variant that is of interest to us here.  For more on Lindstr\"{o}m quantifiers, consult~\cite[Chap.~12]{EbbinghausFlum1999}.  In particular, there are further generalisations of this notion, involving relativisation, vectorisation and taking quotients in the interpretation.  We do not consider these, as we are interested in capturing the hierarchy of quantifiers by their arity, as in~\cite{Hella1996}.

Let $\sigma, \tau$ be signatures
with $\tau = \{R_1,\dots,R_m\}$, and $r_i$ the arity of $R_i$. An interpretation $\mathcal{I}$ of $\tau$ in $\sigma$ with parameters $\tup{z}$ is a tuple of $\sigma$ formulas
 \[(\phi_{R_1}(x_1,\dots,x_{r_1},\tup{z}_1),\dots,\phi_{R_m}(x_1,\dots,x_{r_m},\tup{z}_m))\]
where the $x_i$ are distinct variables, $\tup{z}$ is a
tuple of variables pairwise distinct from the $x$-variables and the $\tup{z}_j$ are (not necessarily distinct) subtuples of $\tup{z}$.\footnote{The requirement that the variables in $\tup{z}$ are distinct from all $x_i$ is as in~\cite{KolaitisVaananen1995} but does not appear in the definition of generalised quantifiers in~\cite{Hella1996} or~\cite{EbbinghausFlum1999}.  We explain the difference this makes in Remark~\ref{rem:typeAtypeB}.\label{fn:gq}}

An interpretation of $\tau$ in $\sigma$ with parameters $\tup{z}$  defines a mapping that takes a $\sigma$-structure $\str{A}$, along with an interpretation $\tup{a}$ of the parameters $\tup{z}$ in $\str{A}$ to a $\tau$-structure $\str{B}$ as follows.   The universe of $\str{B}$ is $B = A$, the same universe as $\str{A}$, and the relations $R_i \in \tau$ are interpreted in $\str{B}$ by $R_i^{\str{B}} = \{(b_1,\dots,b_{r_i}) \in B^{r_i} \mid \str{A}
\models \phi_{R_i}(b_1,\dots,b_{r_i},\tup{a}_i)\}$ where $\tup{a}_i$ is the interpretation of the subtuple $\tup{z}_i$ given by $\tup{a}$.

 Let $L$ be one of the logics in the previous section and $K$ a class of
$\tau$-structures with $\tau = \{R_1,\dots,R_m\}$. The extension $L(Q_{K})$ of $L$ by the \emph{generalised quantifier} (also known as the \emph{Lindstr\"{o}m  quantifier}) for
the class $K$ is obtained by extending the syntax of $L$ by the
following formula formation rule:
\begin{quote}
	Let
$I = \phi_{R_1},\ldots,\phi_{R_m}$
be formulas in $L(Q_{K})$ that form an interpretation of $\tau$ in $\sigma$ with parameters $\tup{z}$. Then $\psi(\tup{z}) = Q_{K}\tup{x}I(\tup{z})$ is a formula in $L(Q_{K})$ over the signature $\sigma$, with the variables in $\tup{x}$ bound.  The semantics of the formula is given by
$(\str{A},\tup{a}) \models Q_{K}\tup{x}I(\tup{z})$, if, and only
if,  $\str{B} := I(\str{A},\tup{a})$ is defined and $\str{B}$ is in $K$.
\end{quote}

Thus, adding the generalised quantifier $Q_K$ to the logic $L$ is
the most direct way to make the class $K$ definable in $L$.
Formally, if $L$ is a regular logic in the sense of~\cite{Ebb85}, then
its extension by $Q_K$ is the minimal regular logic that can also
define $K$.

The classical first-order quantifiers, $\exists$ and $\forall$, can be derived
as generalised quantifiers in the following way. Let $\phi(x, \tup{z})$ be a
formula in $L$. This determines an interpretation into $\tau_1$ the signature
with a single unary relation $U$. The classes $K_{\exists} = \{ \str{A} \
\mid \ U^{\str{A}} \neq \emptyset \}$ and $K_{\forall} = \{ \str{A} \ \mid \
U^{\str{A}} = A \}$ are isomorphism-closed classes of $\tau_1$ structures. Now,
the generalised quantifiers $Q_{K_{\exists}}$ and $Q_{K_{\forall}}$ have the same
formula formation rules as $\exists$ and $\forall$ and the formulas $Q_{K_{\exists}}x\phi(x, \tup{z})$
and $Q_{K_{\forall}}x\phi(x, \tup{z})$ have the same semantics as $\exists x\phi(x, \tup{z})$
and $\forall x\phi(x, \tup{z})$.

\begin{rem}\label{rem:typeAtypeB}
  As noted above (see footnote~\ref{fn:gq}), our definition of the syntax of logics with generalised quantifiers is somewhat non-standard.  We follow Kolaitis and Väänänen~\cite{KolaitisVaananen1995} in requiring that the parameters $\tup{z}$ to an interpretation  $I(\tup{z}) = (\phi_{R_1}(\tup{x}_1,\tup{z}),\ldots, \phi_{R_m}(\tup{x}_m\tup{z}))$ are distinct from the variables in each $\tup{x}_i$.  The standard definition (as in~\cite{Hella1996} and~\cite{EbbinghausFlum1999}) does not require this restriction and so, in a formula $Q_K\tup{x}_1\ldots\tup{x}_m I(\tup{x},\tup{z})$, a variable that appears in $\tup{x}_1$ but not in $\tup{x}_2$ may nonetheless occur free in $\phi_2$ and so be among the parameters $\tup{z}$.  In the absence of any restrictions on the number of variables, this distinction makes no difference to the expressive power of the logic.  However, it does make a difference to how many variables are required in a formula.

  To appreciate the difference, it is worth recalling two equivalent ways in which one can define the fragments $\FO^k$ and $\infL^k$ of first-order logic and infinitary logic respectively.  One can define them syntactically as the restriction of the logics $\FO$ and $\infL$ to formulas using just the variables $x_1,\ldots,x_k$ or more permissively as those formulas in which each subformula has at most $k$ free variables.  It is clear that any formula under the more permissive definition can be transformed, by renaming variables appropriately, to an equivalent one in the restricted syntactic form.  Such a translation is not possible in the presence of generalised quantifiers which may bind variables in more than one subformula simultaneously.  To be precise, a formula of $L(Q_K)$ as we have defined it, and in which no subformula has more than $k$ free variables can be translated to a formula using only the variables $x_1,\ldots,x_k$ only if we use the syntax that permits the quantifier to bind a variable in one subformula while leaving it free in another, i.e.\ the standard syntax.

  Our approach is to adopt the syntactic restriction of Kolaitis and Väänänen, as we have done above, but to define the $k$-variable fragment as consisting of those formulas in which no subformula has more than $k$ free variables, rather than the more restrictive one where no more than $k$ variables appear.  The formal definition follows.
  \end{rem}

  We write $\infL^k(Q_K)$ for the collection of formulas $\phi$ of infinitary logic, extended with the quantifier $Q_K$ such that no subformula of $\phi$ contains more than $k$ free variables.  When we need to refer to the fragment of $\infL^k(Q_K)$ consisting of those formulas with no more than $k$ variables, we will call this ``KV $\infL^k(Q_K)$'' and we may refer to $\infL^k(Q_K)$ as ``Hella $\infL^k(Q_K)$'' to distinguish it.

We define the \emph{arity} of the quantifier $Q_K$ to be $\max\{r_i \mid i \in [m] \}$, i.e.\ the maximum arity of any relation in $\tau$.  Note that this is the number of variables bound by the quantifier.
We write $\infL^{k}(\mathcal{Q}_n)$ for the (Hella) extension of $\infL^k$ with \emph{all} quantifiers of arity $n$.  This is only of interest when $n \leq k$.  Kolaitis and V\"a\"an\"anen~\cite{KolaitisVaananen1995} showed that KV $\infL^{k}(\mathcal{Q}_1)$ is equivalent to $\infC^k$.  However, allowing quantifiers of higher arity gives logics of considerably more expressive power.  In particular, if $\sigma$ is a signature with all relations of arity at most $n$, then \emph{any} property of $\sigma$-structures is expressible in $\infL^{n}(\mathcal{Q}_n)$.  Thus, all properties of graphs, for instance, are expressible in $\infL^{2}(\mathcal{Q}_2)$.

\subsection{Games}
For a pair of structures $\str{A}$ and $\str{B}$ and a logic $L$, we write $\str{A} \Rrightarrow_L \str{B}$ to denote that every sentence of $L$ that is true in $\str{A}$ is also true in $\str{B}$.  When the logic is closed under negation, as is the case with $\FO$ and $\infL$, for instance, $\str{A} \Rrightarrow_L \str{B}$ implies  $\str{B} \Rrightarrow_L \str{A}$.  In this case, we have an equivalence relation between structures and we write $\str{A} \equiv_L \str{B}$.  When $\str{A}$ and $\str{B}$ are finite structures, $\str{A} \Rrightarrow_{\FO} \str{B}$ implies $\str{A} \Rrightarrow_{\infL} \str{B}$, and the same holds for the $k$-variable fragments of these logics (see~\cite{DLW95}).

The relations $\Rrightarrow_L$ are often characterised in terms of games which we generically call \emph{Spoiler-Duplicator} games.  They are played between two players called Spoiler and Duplicator on a board consisting of the two structures $\str{A}$ and $\str{B}$, where the players take turns to place pebbles on elements of the structures.  For instance, the existential-positive $k$-pebble game introduced by Kolaitis and Vardi~\cite{KolaitisVardi1992} which we denote $\exists\game{Peb}{k}{}$ characterises the relation $\Rrightarrow_{\infH^k}$.  In this game, Spoiler and Duplicator each has a collection of $k$ pebbles indexed $1,\ldots,k$.  In each round Spoiler places one of its pebbles on an element of $\str{A}$ and Duplicator responds by placing its corresponding pebble (i.e.\ the one of the same index) on an element of $\str{B}$.  If the partial map taking the element of $\str{A}$ on which Spoiler's pebble $i$ sits to the element of $\str{B}$ on which Duplicator's pebble $i$ is, fails to be a partial homomorphism, then Spoiler has won the game.  Duplicator wins by playing forever without losing.  We get a game characterising $\equiv_{\infL^k}$ if \emph{(i)} Spoiler is allowed to choose, at each move, on which of the two structures it places a pebble and Duplicator is required to respond in the other structure; and \emph{(ii)} Duplicator is required to ensure that the pebbled positions form a partial \emph{isomorphism}.

The equivalence $\equiv_{\infC^k}$ is characterised by the following bijection game which is often attributed to Hella~\cite{Hella1996} but this version actually follows Immerman in Definition 12.22 of~\cite{Immerman1998}.  We write $\game{Bij}{k}{}(\mathcal{A}, \mathcal{B})$ for this bijection game played on $\str{A}$ and $\str{B}$.  Again, there is a set of $k$ pebbles associated with each of the structures $\str{A}$ and $\str{B}$, indexed by the set $[k]$.  At each move, Spoiler chooses an index $i \in [k]$ and Duplicator is required to respond with a bijection $f: A \rightarrow B$.  Spoiler then chooses an element $a \in A$ and pebbles indexed $i$ are placed on $a$ and $f(a)$.  If the partial map defined by the pebbled positions is not a partial isomorphism, then Spoiler has won.  Duplicator wins by playing forever without losing.

The bijection game described above has been widely studied and used to establish that many interesting properties are not invariant under the relation $\equiv_{\infC^k}$ for any $k$.  This is of great interest as these equivalence relations have many natural and independently arising characterisations in algebra, combinatorics, logic and optimisation.  However, in Hella's original work, bijection games appear as a special case of the \emph{$n$-bijective $k$-pebble game}, which we denote $\game{Bij}{n}{k}(\mathcal{A}, \mathcal{B})$ when played on structures $\str{A}$ and $\str{B}$.  This characterises the equivalence relation $\equiv_{\infL^{k}(\mathcal{Q}_n)}$.  Once again, we have a set of $k$ pebbles associated with each of the structures $\str{A}$ and $\str{B}$ and indexed by $[k]$.  At each move, Duplicator is required to give a bijection $f: A \rightarrow B$ and Spoiler chooses a set of up to $n$ pebble indices $p_1,\ldots,p_n \in [k]$ and moves the corresponding indices to elements $a_1,\ldots,a_n \in A$ and $f(a_1),\ldots,f(a_n)$ in $B$. If the partial map defined by the pebbled positions is not a partial isomorphism, then Spoiler has won.  Duplicator wins by playing forever without losing.  Note, in particular, that for Duplicator to have a winning strategy it is necessary that the reducts of $\str{A}$ and $\str{B}$ to relations of arity at most $n$ are isomorphic.  For example, on graphs Spoiler wins any game on non-isomorphic graphs with $n,k \geq 2$.

\begin{rem}\label{rem:games}
It is important to note that $\game{Bij}{n}{k}$ and
$\game{Bij}{k}{}$ differ in the order in which Spoiler
picks up the pebbles to be moved and Duplicator provides
their bijection. Hence, $\game{Bij}{1}{k}$ and $\game{Bij}{k}{}$ are in fact different games, though this difference is
often overlooked in the literature.

Similarly, we could define a version of the bijection game where Spoiler first picks up $n$ pebbles and then Duplicator provides a bijection.  This would correspond to equivalence in the logic KV $\infL^{k}(\mathcal{Q}_n)$.
\end{rem}

\subsection{Comonads}
We assume that the reader is familiar with basic definitions from category theory, in particular the notions of category, functor and natural transformation.  An introduction may be found in~\cite{AbramskyT2010}.  For a finite signature $\sigma$, we are interested in the category $\mathcal{R}(\sigma)$ of relational structures over $\sigma$.  The objects of the category are such structures and the maps are homomorphisms between structures.

A \emph{comonad} $\mathbb{T}$ on a category $\mathcal{C}$ is a triple $(\mathbb{T},\epsilon,\delta)$ where $\mathbb{T}$ is an endofunctor of $\mathcal{C}$, and $\epsilon$ and $\delta$ are natural transformations, giving for each object $A \in \mathcal{C}$, morphisms $\epsilon_A : \mathbb{T} A \rTo A$ and  $\delta_A : \mathbb{T} A \rTo \mathbb{T} \mathbb{T} A$ so that the following diagrams commute.
\[\begin{diagram}
\mathbb{T} A & \rTo^{\delta_A} & \mathbb{T} \mathbb{T} A \\
\dTo^{\delta_A} & & \dTo_{\mathbb{T} \delta_A} \\
\mathbb{T} \mathbb{T} A & \rTo_{\delta_{\mathbb{T} A}} & \mathbb{T} \mathbb{T} \mathbb{T} A
\end{diagram} \qquad \qquad
\begin{diagram}
\mathbb{T} A & \rTo^{\delta_A} & \mathbb{T} \mathbb{T} A \\
\dTo^{\delta_{\str{A}}} & \rdEq & \dTo_{\mathbb{T} \epsilon_{\str{A}}} \\
\mathbb{T} \mathbb{T} A & \rTo_{\epsilon_T} & \mathbb{T} A
\end{diagram}\]

We call $\epsilon$ the \emph{counit} and $\delta$ the comultiplication of the comonad $(\mathbb{T},\epsilon,\delta)$.

Associated with any comonad $(\mathbb{T},\epsilon,\delta)$ is a \emph{Kleisli category} we denote $\mathcal{K}(\mathbb{T})$.  The objects are the objects of the underlying category $\mathcal{C}$ and the maps $A \rTo_{\mathcal{K}(\mathbb{T})} B$ are morphisms $\mathbb{T} A \rTo B$ in $\mathcal{C}$.  Composition is given by the comultiplication:
\begin{diagram}
\mathbb{T} \str{A} & \rTo^{\delta_\str{A}} & \mathbb{T} \mathbb{T} \str{A} & \rTo^{\mathbb{T} f} & \mathbb{T} \str{B} & \rTo^{g} & \str{C}.
\end{diagram}
The identity morphisms are given by the counit: $\epsilon_{\str{A}} : \mathbb{T} \str{A} \rTo \str{A}$.

A coalgebra for the comonad is a map $\alpha : A \rightarrow \mathbb{T} A$ such that the following diagrams commute.

\[\begin{diagram}
A & \rTo^{\alpha} & \mathbb{T} A \\
\dTo^{\alpha} & & \dTo_{\delta_{\str{A}}} \\
\mathbb{T} A & \rTo_{\mathbb{T} \alpha} & \mathbb{T} \mathbb{T} A
\end{diagram}  \qquad \qquad
\begin{diagram}
A & \rTo^{\alpha} & \mathbb{T} A \\
& \rdTo_{\mathrm{id}_{\str{A}}} & \dTo_{\epsilon_{\str{A}}} \\
& & A
\end{diagram}\]

Abramsky et al.~\cite{Abramsky2017} describe the construction of a comonad $\T{k}$, graded by $k$, on the category $\mathcal{R}(\sigma)$ which exposes an interesting relationship between the games $\exists\game{Peb}{k}{}(\mathcal{A}, \mathcal{B})$ and $\game{Bij}{k}{}(\mathcal{A}, \mathcal{B})$.  Specifically, it shows that Duplicator winning strategies in the latter are exactly the isomorphisms in a category in which the morphisms are winning strategies in the former.

For any $\str{A}$, $\T{k}\str{A}$ is an infinite structure (even when $\str{A}$ is finite) with universe $(A \times [k])^{+}$.  The counit $\epsilon_{\str{A}}$ takes a sequence $[(a_1,p_1),\ldots,(a_m,p_m)]$ to $a_m$, i.e.\ the first component of the last element of the sequence.  The comultiplication $\delta_{\str{A}}$ takes a sequence $[(a_1,p_1),\ldots,(a_m,p_m)]$ to the sequence $[(s_1,p_1),\ldots,(s_m,p_m)]$ where $s_i = [(a_1,p_1),\ldots,(a_i,p_i)]$.
The relations are defined so that $(s_1,\ldots,s_r) \in R^{\T{k}\str{A}}$ if, and only if, the $s_i$ are all comparable in the prefix order of sequences (and hence form a chain), $R^{\str{A}}(\epsilon_{\str{A}}(s_1),\ldots,\epsilon_{\str{A}}(s_r))$ and whenever $s_i$ is a prefix of $s_j$ and ends with the pair $(a,p)$, there is no prefix of $s_j$ properly extending $s_i$ which ends with $(a',p)$ for any $a' \in A$.

It is convenient to consider structures over a signature $\sigma \cup \{I\}$ where $I$ is a new binary relation symbol.  An $I$-structure is a structure over this signature which interprets $I$ as the identity relation.  Note that even when $\str{A}$ is an $I$-structure, $\T{k}\str{A}$ is not one.  The key results from~\cite{Abramsky2017} relating the comonad with pebble games can now be stated as establishing a precise translation between \emph{(i)} morphisms $\str{A} \rTo_{\mathcal{K}(\T{k})} \str{B}$ for $I$-structures $\str{A}$ and $\str{B}$; and \emph{(ii)} winning strategies for Duplicator in $\exists\game{Peb}{k}{}(\mathcal{A}, \mathcal{B})$; and similarly a precise translation between \emph{(i)} isomorphisms in $\mathcal{K}(\T{k})$ between $\str{A}$ and  $\str{B}$ for $I$-structures $\str{A}$ and $\str{B}$; and \emph{(ii)} winning strategies for Duplicator in $\game{Bij}{k}{}(\mathcal{A}, \mathcal{B})$.

A key result from the construction of the comonad $\T{k}$ is the relationship between the coalgebras of this comonad and tree decompositions.  In particular, a structure $\str{A}$ has a coalgebra $\alpha: \str{A} \rightarrow \T{k} \str{A}$ if, and only if, the treewidth of $\str{A}$ is at most $k-1$.  This relationship between coalgebras and tree decompositions is established through a definition of a \emph{tree traversal} which we review in Section~\ref{sec:coalgebra} below.

\section{Games and Logic with Generalised Quantifiers}\label{sec:logic}

The $n$-bijective $k$-pebble game $\game{Bij}{n}{k}$ as introduced by Hella is a model-comparison game which captures equivalence of structures over the logic $\infL^{k}(\mathcal{Q}_n)$, i.e.\ $k$-variable infinitary logic where the allowed quantifiers are all generalised quantifiers with arity $\leq n$. This game generalises a variant of the bijection game $\game{Bij}{k}{}$ which captures equivalence over $\mathcal{C}^{k}$, $k$-variable infinitary logic with counting quantifiers (which is equivalent to a fragment of $\infL^{k}(\mathcal{Q}_1)$ as shown by Kolaitis and V\"a\"an\"anen~\cite{KolaitisVaananen1995}). In this section, we introduce a family of games which relax the rules of $\game{Bij}{n}{k}$ and show their correspondence to different fragments of $\infL^{k}(\mathcal{Q}_n)$. In particular, we introduce a ``one-way'' version of $\game{Bij}{n}{k}$ which is crucial to our construction of a modified version of the $\T{k}$ comonad for these games.

\subsection{\texorpdfstring{Relaxing $\game{Bij}{n}{k}$}{Relaxing Bij\textsuperscript{k}\textsubscript{n}}}\label{sec:games}

Recall that each round of $\game{Bij}{n}{k}(\str{A}, \str{B})$ involves Duplicator selecting a bijection $f: A \rightarrow B$ and ends with a test of whether for the pebbled positions $(a_i, b_i)_{i\in[k]}$ it is the case that for any $\{i_1, \ldots i_r\}\subset[k]$  \[(a_{i_1}, \dots a_{i_r}) \in R^{\str{A}} \iff (b_{i_1}, \dots b_{i_r}) \in R^{\str{B}}\] where Duplicator loses if the test is failed. For the rest of the round, Spoiler rearranges up to $n$ pebbles on $\str{A}$ with the corresponding pebbles on $\str{B}$ moved according to $f$.

So, to create from $\game{Bij}{n}{k}$ a ``one-way'' game from $\str{A}$ to $\str{B}$ we need to relax the condition that $f$ be a bijection and the $\iff$ in the final test. The following definition captures the most basic such relaxation:

\begin{defi}\label{defn:function-game}

  For two relational structures $\mathcal{A}$, $\mathcal{B}$, the positive $k$-pebble $n$-function game, $+\game{Fun}{n}{k}(\mathcal{A}, \mathcal{B})$ is played by Spoiler and Duplicator. Prior to the $j$th round the position consists of partial maps $\pi_{j-1}^a:[k] \pra A$ and $\pi^b_{j-1}:[k] \pra B$.  In Round $j$
\begin{itemize}
  \item Duplicator provides a function $h_j: A \rightarrow B$ such that for each $i \in [k]$, $h_j(\pi^a_{j-1}(i)) = \pi_{j-1}^b(i)$.
  \item Spoiler picks up to $n$ distinct \emph{pebbles}, i.e.\ elements $p_1, \dots p_m \in [k] (m \leq n)$ and $m$ elements $x_1, \dots x_m \in A$.
  \item  The updated position is given by $\pi^a_j(p_l) = x_l$ and $\pi^b_j(p_l) = h_j(x_l)$ for $l \in[m]$; and $\pi_j^a(i) = \pi_{j-1}^a(i)$ and $\pi_j^b(i) = \pi_{j-1}^b(i)$ for $i \not\in \{p_1,\ldots,p_m\}$.
  \item If there is some $R \in \sigma$ and $(i_1, \dots i_{r})\in [k]^{r}$ with $(\pi^a_j(i_1),\ldots,\pi^a_j(i_r)) \in  R^{\str{A}}$ \\
     but $(\pi^b_j(i_1),\ldots,\pi^b_j(i_r)) \not\in  R^{\str{B}}$, then Spoiler has won the game.
\end{itemize}
Duplicator wins by preventing Spoiler from winning.
\end{defi}

As this game is to serve as the appropriate one-way game
for $\game{Bij}{n}{k}$, it is worth asking how this this
game relates to $\exists\game{Peb}{k}{}$ (the one-way game
for $\game{Bij}{k}{}$) which makes no mention of functions
in its definition. The answer comes in recalling Abramsky
et al.'s presentation of a (deterministic) strategy for
Duplicator in $\exists\game{Peb}{k}{}(\mathcal{A}, \mathcal{B})$
as a collection of \textit{branch maps} $\phi_{s, i}: A
\rightarrow B$ for each $s \in (A \times [k])^{\ast}$, a history
of Spoiler moves and $i \in [k]$ a pebble index. These branch
maps tell us how Duplicator would respond to Spoiler moving pebble $i$ to any element in $A$ given the moves $s$ that Spoiler has played in preceding rounds and can be thought of as a function which Duplicator provides to Spoiler after Spoiler has indicated which pebble he will move. In the game in Definition~\ref{defn:function-game}, Duplicator provides this function before Spoiler indicates which
pebbles are to be moved.

In addition to this game, we now define some other relaxations of $\game{Bij}{n}{k}$ which are important. In particular we define the following \textit{positive} games by retaining that the pebbled position need only preserve positive atoms at the end of each round but varying the condition on $f$.

\begin{defi}
 For two relational structures $\mathcal{A}$, $\mathcal{B}$, the positive $k$-pebble $n$-injection (resp.\ surjection, bijection) game, $+\game{Inj}{n}{k}(\mathcal{A}, \mathcal{B})$ (resp.\ $+\game{Surj}{n}{k}(\mathcal{A}, \mathcal{B})$, $+\game{Bij}{n}{k}(\mathcal{A}, \mathcal{B})$) is played by Spoiler and Duplicator.  Prior to the $j$th round the position consists of partial maps $\pi_{j-1}^a:[k] \pra A$ and $\pi^b_{j-1}:[k] \pra B$.  In Round $j$
\begin{itemize}
  \item Duplicator provides an injection (resp.\ a surjection, bijection) $h_j: A \rightarrow B$ such that for each $i \in [k]$, $h_j(\pi^a_{j-1}(i)) = \pi_{j-1}^b(i)$.
  \item Spoiler picks up to $n$ distinct \emph{pebbles}, i.e.\ elements $p_1, \dots p_m \in [k] (m \leq n)$ and $m$ elements $x_1, \dots x_m \in A$.
  \item  The updated position is given by $\pi^a_j(p_l) = x_l$ and $\pi^b_j(p_l) = h_j(x_l)$ for $l \in[m]$; and $\pi_j^a(i) = \pi_{j-1}^a(i)$ and $\pi_j^b(i) = \pi_{j-1}^b(i)$ for $i \not\in \{p_1,\ldots,p_m\}$.
  \item Spoiler has won the game if there is some $R \in \sigma$ and $(i_1, \dots i_{r})\in [k]^{r}$ \\
    with $(\pi^a_j(i_1),\ldots,\pi^a_j(i_r)) \in  R^{\str{A}}$ but $(\pi^b_j(i_1),\ldots,\pi^b_j(i_r)) \not\in  R^{\str{B}}$.
\end{itemize}
Duplicator wins by preventing Spoiler from winning.

\end{defi}

Strengthening the test condition in each round so that Spoiler wins if there is some $R \in \sigma$ and $(i_1, \dots i_{r})\in [k]^{r}$ with $(\pi^a_j(i_1),\ldots,\pi^a_j(i_r)) \in  R^{\str{A}}$ if, and only if,  $(\pi^b_j(i_1),\ldots,\pi^b_j(i_r)) \not\in  R^{\str{B}}$,
we get the definitions for the games $\game{Fun}{n}{k}$, $\game{Inj}{n}{k}$, $\game{Surj}{n}{k}$ and $\game{Bij}{n}{k}$ where the latter is precisely the $n$-bijective $k$-pebble game of Hella.  We recap the poset of the games we've just defined ordered by strengthening of the rules/restrictions on Duplicator in the Hasse diagram in Figure~\ref{fig:game-Hasse}.  Here a game $\mathcal{G}$ is above $\mathcal{G}'$ if a Duplicator winning strategy in $\mathcal{G}$ is also one in $\mathcal{G}'$.
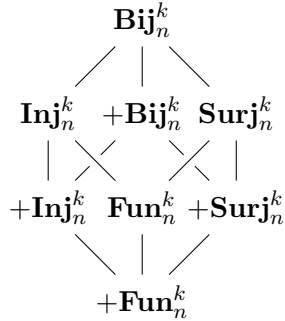
\begin{figure}[h]
  \centering

\begin{tikzpicture}
  \node (max) at (0,2.5) {$\game{Bij}{n}{k}$};
  \node (a) at (-1.25,1.25) {$\game{Inj}{n}{k}$};
  \node (b) at (0,1.25) {$+\game{Bij}{n}{k}$};
  \node (c) at (1.25,1.25) {$\game{Surj}{n}{k}$};
  \node (d) at (-1.25,0) {$+\game{Inj}{n}{k}$};
  \node (e) at (0,0) {$\game{Fun}{n}{k}$};
  \node (f) at (1.25,0) {$+\game{Surj}{n}{k}$};
  \node (min) at (0,-1.25) {$+\game{Fun}{n}{k}$};
  \draw (min) -- (d) -- (a) -- (max) -- (b) -- (d)
  (e) -- (min) -- (f) -- (c) -- (max)
  (f) -- (b);
  \draw[preaction={draw=white, -,line width=6pt}] (a) -- (e) -- (c);
\end{tikzpicture}

  \caption{Hasse Diagram of Games}\label{fig:game-Hasse}
\end{figure}

\subsection{Logics with generalised quantifiers}\label{sec:logics}

In Section~\ref{sec:background}, we introduce for each $n, k \in \nats$ the logics, $\infL^{k}(\mathcal{Q}_n)$ as the infinitary logic extended with all generalised quantifiers of arity $n$.

In this section we explore fragments of $\infL^{k}(\mathcal{Q}_n)$ defined by restricted classes of generalised quantifiers, which we introduce next.

\begin{defi}
  A class of $\sigma$-structures $K$ is homomorphism-closed if for all homomorphisms $f: \str{A} \rightarrow \str{B}$
\[ \str{A} \in K \implies \str{B} \in K.\]
  Similarly, we say $K$ is injection-closed (resp.\ surjection-closed, bijection-closed) if for all injective homomorphisms (resp.\ surjective, bijective homomorphisms) $f: \str{A} \rightarrow \str{B}$
\[ \str{A} \in K \implies \str{B} \in K.\]

  We write $\mathcal{Q}^{\text{h}}_n$ for the class of all generalised quantifiers $Q_K$ of arity $n$ where $K$ is homomorphism-closed.  Similarly,  we write $\mathcal{Q}^{\text{i}}_n$, $\mathcal{Q}^{\text{s}}_n$ and $\mathcal{Q}^{\text{b}}_n$ for the collections of $n$-ary quantifiers based on injection-closed, surjection-closed and bijection-closed classes.
\end{defi}

In order to define logics which incorporate these restricted classes of quantifiers, we first define a base logic without quantifiers or negation.
\begin{defi}
Fix a signature $\sigma$.

We denote by $\posqfk{k}[\sigma]$, the class of positive infinitary $k$-variable quantifier-free formulas over $\sigma$.  That means the $k$ variable fragment of the class of formulas
\[ \posqfk{}[\sigma] :\defeq R(x_1, \dots x_m) \ | \ \bigwedge_{\mathcal{I}} \phi \ | \ \bigvee_{\mathcal{J}} \psi \]
for any $R \in \sigma$.
We use $\qfk{k}[\sigma]$ to denote a similar class of formulas but with negation permitted on atoms.
\end{defi}

This basic set of formulas can be extended into a logic by adding some set of quantifiers as described here:

\begin{defi}
For $\mathcal{Q}$ some collection of generalised quantifiers, we denote by $\posqfk{k}(\mathcal{Q})$ the smallest extension of $\posqfk{k}$ closed under the construction
\[ \mathbf{Q} x_1, \dots x_n. \ (\psi_T(\mathbf{x}_{T},\mathbf{y}_T))_{T \in \tau}  \]
for any $\mathbf{Q} \in \mathcal{Q}$. $\qfk{k}(\mathcal{Q})$ is the same logic but with negation on atoms.
Note that $\infH^{k} \equiv \posqfk{k}(\exists)$ and, as we can always push negation down to the level of atoms in $\infL^{k}$, $\infL^{k} \equiv \qfk{k}(\exists, \forall)$.
\end{defi}

With this definition we are ready to introduce our logics. These are $\qfk{k}(\mathcal{Q}^{\text{h}}_n)$, $\qfk{k}(\mathcal{Q}^{\text{i}}_n)$, $\qfk{k}(\mathcal{Q}^{\text{s}}_n)$ and $\qfk{k}(\mathcal{Q}^{\text{b}}_n)$ and their positive counterparts  $\posqfk{k}(\mathcal{Q}^{\text{h}}_n)$, $\posqfk{k}(\mathcal{Q}^{\text{i}}_n)$, $\posqfk{k}(\mathcal{Q}^{\text{s}}_n)$ and $\posqfk{k}(\mathcal{Q}^{\text{b}}_n)$.
  The obvious inclusion relationships between these logics are given by the Hasse diagram in Figure~\ref{fig:logic-Hasse}.  As we shall see, these logics are governed exactly by the games pictured in Figure~\ref{fig:game-Hasse}.

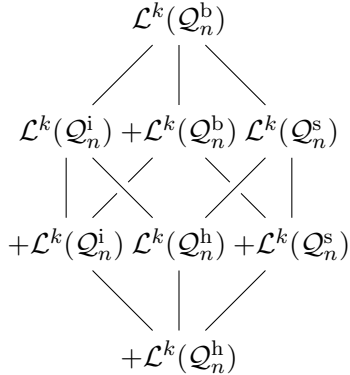
\begin{figure}[h]
  \centering
\[
\begin{tikzpicture}
  \node (max) at (0,3) {$\qfk{k}(\mathcal{Q}^{\text{b}}_n)$};
  \node (a) at (-1.5,1.5) {$\qfk{k}(\mathcal{Q}^{\text{i}}_n)$};
  \node (b) at (0,1.5) {$\posqfk{k}(\mathcal{Q}^{\text{b}}_n)$};
  \node (c) at (1.5,1.5) {$\qfk{k}(\mathcal{Q}^{\text{s}}_n)$};
  \node (d) at (-1.5,0) {$\posqfk{k}(\mathcal{Q}^{\text{i}}_n)$};
  \node (e) at (0,0) {$\qfk{k}(\mathcal{Q}^{\text{h}}_n)$};
  \node (f) at (1.5,0) {$\posqfk{k}(\mathcal{Q}^{\text{s}}_n)$};
  \node (min) at (0,-1.5) {$\posqfk{k}(\mathcal{Q}^{\text{h}}_n)$};
  \draw (min) -- (d) -- (a) -- (max) -- (b) -- (d)
  (e) -- (min) -- (f) -- (c) -- (max)
  (f) -- (b);
  \draw[preaction={draw=white, -,line width=6pt}] (a) -- (e) -- (c);
\end{tikzpicture}
\]
 \caption{Hasse Diagram of Logics}\label{fig:logic-Hasse}
\end{figure}

Before we prove the correspondence with the aforementioned games, we highlight two important facts about this family of logics. Firstly, we show that $\qfk{k}(\mathcal{Q}^{\text{b}}_n)$ is equivalent to Hella's original infinitary logic with $n$-ary generalised quantifiers and, secondly, we show how these families of generalised quantifiers relate the sizes of structures.

\subsubsection{\texorpdfstring{$\qfk{k}(\mathcal{Q}^{\text{b}}_n)$ and $\infL^{k}(\mathcal{Q}_n)$ are equivalent}{L\textsuperscript{k}(Q\textsuperscript{b}\textsubscript{n}) and L\textsuperscript{k}\textsubscript{∞}(Q\textsubscript{n}) are equivalent}}

Theorem~\ref{thm:logic} proves, among other things, that Duplicator has a winning strategy in the game $\game{Bij}{n}{k}(\str{A}, \str{B})$ if, and only if, $\str{A} \equiv_{\qfk{k}(\mathcal{Q}^{\text{b}}_n)} \str{B}$. However, Hella~\cite{Hella1989} originally characterised such pairs of structures by equivalence in the seemingly more powerful logic $\infL^{k}(\mathcal{Q}_n)$. Here, we show from first principles that these two logics are indeed equivalent.

We say that two logics $\mathcal{L}_0$ and $\mathcal{L}_1$ are \emph{equivalent} if for every signature $\sigma$ and every formula  $\phi(\tup{y}) \in \mathcal{L}_{i}[\sigma]$ there exists an equivalent formula $\overline{\phi}(\tup{y}) \in \mathcal{L}_{1-i}[\sigma]$ such that for any $\sigma$-structure $\mathcal{A}$ and any tuple of elements $\tup{a}$ with the same length as $\tup{y}$, $\mathcal{A}, \tup{a} \models \phi(\tup{y})$ if and only if $\mathcal{A}, \tup{a} \models \overline{\phi}(\tup{y})$. For two such equivalent logics we will write $\mathcal{L}_0 \equiv \mathcal{L}_1$.

To show that $\qfk{k}(\mathcal{Q}^{\text{b}}_n) \equiv \infL^{k}(\mathcal{Q}_n)$ for any $n$ and $k$ we need to overcome two differences between these logics. Firstly, the class $\mathcal{Q}^{\text{b}}_n$ of bijective-homomorphism-closed $n$-ary quantifiers is a proper subclass of $\mathcal{Q}_n$ of all isomorphism-closed $n$-ary quantifiers. The following observation provides a way of replacing general isomorphism-closed classes with bijective-homomorphism-closed ones by modifying the signature.

\begin{obs}\label{obs:bijisom}
For $K$ an isomorphism-closed class of $\tau$-structures, if $\tau^{\prime} = \tau \cup \{\overline{R} \ | \ R \in \tau \}$ then
\[K^{\prime} = \{ \str{A} \in \strs{\tau^{\prime}} \ | \ \langle A , (R^{\str{A}})_{R \in \tau} \rangle \in K \text{ and } \forall R \in \tau, \ \overline{R}^{\str{A}} = A^{\text{arity}(R)} \setminus R^{\str{A}}\} \]
 is a bijective-homomorphism closed class of $\tau^{\prime}$ structures.
\end{obs}

An important consequence of this is that any such $K$, the formula \[\phi(\tup{y}) = \mathcal{Q}_K x_1, \dots x_n\ . \ (\psi_R(\mathbf{x}_R, \mathbf{y}_R))_{R \in \tau}\] is equivalent to the formula \[\phi'(\tup{y}) = \mathcal{Q}_{K'} x_1, \dots x_n\ . \ (\psi'_R(\mathbf{x}_R, \mathbf{y}_R))_{R \in \tau'},\] where for any $R \in \tau$ $\psi'_R = \psi_R $ and $\psi'_{\overline{R}} = \neg\psi_R $.

The second difference between these two logics is the  role of negation. As defined in this section, $\qfk{k}(\mathcal{Q}^{\text{b}}_n)$ only allows negation on atoms, whereas $\infL^{k}(\mathcal{Q}_n)$ allows negation throughout formulas. The following observation is important for dealing with this difference.

\begin{obs}\label{obs:querycomp}
A class of $\tau$-structures $K$ is isomorphism-closed if, and only if, its complement $K^{\text{c}}$ is.
\end{obs}

This implies that the formula $\phi(\tup{y}) = \neg\mathcal{Q}_K x_1, \dots x_n\ . \ (\psi_R(\mathbf{x}_R, \mathbf{y}_R))_{R \in \tau}$ is equivalent to $\phi'(\tup{y}) = \mathcal{Q}_{K^{\text{c}}} x_1, \dots x_n\ . \ (\psi_R(\mathbf{x}_R, \mathbf{y}_R))_{R \in \tau}$.

We are now ready to prove the desired equivalence of logics.
\begin{lem}
  For all $n, k \in \nats$, $\qfk{k}(\mathcal{Q}^{\text{b}}_n) \equiv \infL^{k}(\mathcal{Q}_n)$.
\end{lem}

\begin{proof}
Clearly $\qfk{k}(\mathcal{Q}^{\text{b}}_n)$ is contained in $\infL^{k}(\mathcal{Q}_n)$, so we focus on translating a formula $\phi(\tup{y}) \in \infL^{k}(\mathcal{Q}_n)$ to an equivalent $\tilde{\phi}(\tup{y})$ in $\qfk{k}(\mathcal{Q}^{\text{b}}_n)$. This can be done by induction on the quantifier depth of $\phi$. For quantifier depth $0$, there are no quantifiers to be replaced and any negation is either on atoms or can be assumed to be on atoms by appropriately distributing over conjunction or disjunction.

Now we assume $\phi$ has quantifier depth $q$. Without loss of generality, we can assume that $\phi$ is of the form $\mathcal{Q}_K x_1, \dots x_n\ . \ (\psi_R(\mathbf{x}_R, \mathbf{y}_R))_{R \in \tau}$ for some isomorphism-closed class $K$ of $\tau$-structures. Indeed, if $\phi$ contains a leading negation we can use Observation~\ref{obs:querycomp} to remove the negation by replacing $K$ with $K^{\text{c}}$. Note that the formulas $\psi_R$ and $\neg\psi_R$ have quantifier depth strictly less than $q$ and so by induction they have equivalents $\tilde{\psi_R}$ and $\tilde{\neg\psi_R}$ in $\qfk{k}(\mathcal{Q}^{\text{b}}_n)$. Now, using the consequence of Observation~\ref{obs:bijisom} mentioned above, we can define $\tilde{\phi}$ as $\mathcal{Q}_{K'} x_1, \dots x_n\ . \ (\tilde{\psi'_R}(\mathbf{x}_R, \mathbf{y}_R))_{R \in \tau'}$
\end{proof}

\subsubsection{Generalised quantifiers and size}
For any relational signature $\sigma$ let $\strs{\sigma}^{=M}$ denote the collection of $\sigma$-structures whose universe has exactly $M$ elements.  Let $\strs{\sigma}^{\geq M} = \bigcup_{m \geq M}\strs{\sigma}^{=m}$ and similarly $\strs{\sigma}^{\leq M} = \bigcup_{m \leq M}\strs{\sigma}^{=m}$.  It is obvious that $\strs{\sigma}^{=M}$ is bijection-closed, $\strs{\sigma}^{\geq M}$ is injection-closed and $\strs{\sigma}^{\leq M}$ is surjection-closed.
When $\sigma = \emptyset$ is the empty signature this gives us classes of sets $\mathcal{K}^{=M}$, $\mathcal{K}^{ \geq M}$ and $\mathcal{K}^{\leq M}$ which are closed under bijections, injections and surjections respectively.
As any signature $\sigma$ admits an empty interpretation into the empty signature which sends
any $\sigma$-structure to its underlying set, we can create sentences $B_m, I_m,$ and $ S_m$
by binding the nullary quantifier $\mathcal{Q}_{K}$, for $K = $ $\mathcal{K}^{=M}$, $\mathcal{K}^{ \geq M}$ and $\mathcal{K}^{\leq M}$ respectively, to this empty interpretation. As noted in the following observation
these sentences are important for comparing the sizes of structures, in any signature.
\begin{obs}\label{obs:size}
  For all $n,k, m \in\nats$ there are sentences $B_m, I_m,$ and $ S_m$ in $\posqfk{k}(\mathcal{Q}^b_n)$, $\posqfk{k}(\mathcal{Q}^i_n),$ and $ \posqfk{k}(\mathcal{Q}^s_n)$ respectively, such that
  \begin{align*}
    \str{A} \models B_m &\iff \abs{A} = m \\
    \str{A} \models I_m &\iff \abs{A} \geq m \\
    \str{A} \models S_m &\iff \abs{A} \leq m
  \end{align*}

  As a direct result of this we have that
  \begin{align*}
    \str{A} \Rrightarrow_{\posqfk{k}(\mathcal{Q}^b_n)} \str{B} &\implies \abs{A} = \abs{B} \\
    \str{A} \Rrightarrow_{\posqfk{k}(\mathcal{Q}^i_n)} \str{B} &\implies \abs{A} \leq \abs{B} \\
    \str{A} \Rrightarrow_{\posqfk{k}(\mathcal{Q}^s_n)} \str{B} &\implies \abs{A} \geq \abs{B}.
  \end{align*}
\end{obs}

\subsection{Games and logics correspond}
So far we have introduced a series of games and logics which are all variations on Hella's $n$-bijection $k$-pebble game, $\game{Bij}{n}{k}$, and the corresponding logic $\infL^{k}(\mathcal{Q}_n)$. Here we show that these games and logics match up in the way that one would expect from looking at the respective refinement posets in Figures~\ref{fig:game-Hasse} and~\ref{fig:logic-Hasse}.

In order to present the proof of this in a uniform fashion, we label the corners of these cubes by three parameters $\injpar, \surjpar, \negpar \in \{0,1\}$ as indicated in Figure~\ref{fig:cube}. These
parameters signal the presence or absence of certain rules in the corresponding
games. In particular, $\injpar$ and $\surjpar$ indicate if the function provided by Duplicator
in each round is required to be injective or surjective respectively and $\negpar$
indicates if Spoiler wins when negated atoms are not preserved by the partial
map defined at the end of a round.

\begin{figure}[h]
  \centering
\[
\begin{tikzpicture}
  \node (max) at (0,2.5) {$(1,1,1)$};
  \node (a) at (-1.25,1.25) {$(1,0,1)$};
  \node (b) at (0,1.25) {$(1,1,0)$};
  \node (c) at (1.25,1.25) {$(0,1,1)$};
  \node (d) at (-1.25,0) {$(1,0,0)$};
  \node (e) at (0,0) {$(0,0,1)$};
  \node (f) at (1.25,0) {$(0,1,0)$};
  \node (min) at (0,-1.25) {$(0,0,0)$};
  \draw (min) -- (d) -- (a) -- (max) -- (b) -- (d)
  (e) -- (min) -- (f) -- (c) -- (max)
  (f) -- (b);
  \draw[preaction={draw=white, -,line width=6pt}] (a) -- (e) -- (c);
\end{tikzpicture}
\]
\caption{Cube of parameters}\label{fig:cube}
\end{figure}
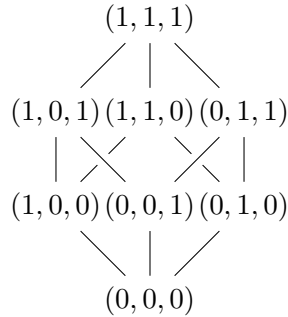

Now we define the aliases of each of the games which modify $\game{Fun}{n}{k}$ as follows, with the games defined lining up with the games defined in Section~\ref{sec:games}.

\begin{defi}
  For two $\sigma$-structures $\str{A}$ and $\str{B}$, the game $(\injpar, \surjpar, \negpar)$-$\game{Fun}{n}{k}(\str{A}, \str{B})$ is played by Spoiler and Duplicator in the same fashion as the game $\game{Fun}{n}{k}(\str{A}, \str{B})$ with the following additional rules:
  \begin{enumerate}
  \item When Duplicator provides a function $f: A \rightarrow B$ at the beginning of a round, $f$ is required to be
  \begin{itemize}
    \item injective if $\injpar = 1$ and
    \item surjective if $\surjpar = 1$.
  \end{itemize}
\item If $\negpar = 1$, Spoiler wins at move $j$ if the partial map taking $\pi^a_j(i)$ to $\pi^b_j(i)$ fails to preserve negated atoms as well as atoms.
  \end{enumerate}
\end{defi}

Similarly, we define parameterised aliases for the logics introduced in Section~\ref{sec:logics}. To lighten our notational burden, we use $\eplnkinftyomega{n,k}$ to denote the logic $\posqfk{k}(\mathcal{Q}^{\text{h}}_n)$ throughout this section.

\begin{defi}
We define $\eplnkinftyomega{n,k}_{\mathbf{x}}$ to be the logic $\eplnkinftyomega{n,k}$ extended by
\begin{enumerate}
 \item all $n$-ary generalised quantifiers closed by all homomorphisms which are:
 \begin{itemize}
   \item injective, if $\injpar =1$; and
   \item surjective, if $\surjpar = 1$
 \end{itemize}
 \item if $\negpar = 1$, negation on atoms.
\end{enumerate}
\end{defi}
For example, $\eplnkinftyomega{n,k}_{001}$ extends $\eplnkinftyomega{n,k}$ with negation on atoms but contains no additional quantifiers as all $n$-ary quantifiers closed under homomorphisms are already in $\eplnkinftyomega{n,k}$.  On the other hand, $\eplnkinftyomega{n,k}_{110}$ does not allow negation on atoms but allows all quantifiers that are closed under bijective homomorphisms.

Now to prove the desired correspondence between $\mathbf{x}$-$\game{Fun}{n}{k}$ and $\eplnkinftyomega{n,k}_{\mathbf{x}}$, we adapt a proof from Hella~\cite{Hella1996} to work for this parameterised set of games.

For this we need the language of \emph{forth} systems which are used
as an explicit representation of a Duplicator winning strategy\footnote{These are called ``$k$-variable $n$-bijective back-and-forth sets'' in Hella's paper, where the ``back'' condition is implicit in the use of bijections. We drop that in the present generalisation.}.  We provide the appropriate generalised definition here:
\begin{defi}
Let $\textbf{Part}^{k}_{\negpar}(\str{A}, \str{B})$ be the set of all partial functions $A \rightharpoonup B$ which preserve atoms (i.e.\ are partial homomorphisms) and, if $\negpar = 1$ additionally preserve negated atoms.

A set $\mathcal{S}\subset \textbf{Part}^{k}_{\negpar}(\str{A},\str{B})$ is a forth system for the game $(\injpar, \surjpar, \negpar)\text{-}\game{Fun}{n}{k}(\str{A}, \str{B})$ if it satisfies the following properties:
\begin{itemize}
    \item \textbf{Downwards closure}: If $f \in \mathcal{S}$ then $g \in \mathcal{S}$ for any $g \subset f$
    \item \textbf{($\injpar, \surjpar$)-forth property} For any $f$ in $\mathcal{S}$ s.t.\ $\abs{f} \leq k$, there exists a function $\phi_f : A \rightarrow B$, which is injective if $\injpar = 1$ and surjective if $\surjpar = 1$ s.t.\ for every $C \subset \text{dom}(f), D \subset A$ with $|D| \leq n$ and $\abs{C \cup D} \leq k$
    we have $(f \downharpoonright C) \cup (\phi_f \downharpoonright D) \in \mathcal{S}$.
\end{itemize}

\end{defi}
Note that in the ``forth'' condition, there is a single function $\phi_f$ that yields the property for any choice of set $C$.  This captures the condition in the game where Duplicator has to play this function before Spoiler chooses which pebbles to move (cf.~Remark~\ref{rem:games}).
As this definition is essentially an unravelling of a Duplicator winning strategy for the game $(\injpar, \surjpar, \negpar)\text{-}\game{Fun}{n}{k}(\str{A}, \str{B})$ we get the following.

\begin{lem}
  There is a forth system $\mathcal{S}$ containing the empty partial homomorphism $\emptyset$ if, and only if, Duplicator has a winning strategy for the game $(\injpar, \surjpar, \negpar)\text{-}\game{Fun}{n}{k}(\str{A}, \str{B})$
\end{lem}
\begin{proof}
  For the forward direction we note that if the pebbled position at the beginning of some round of $(\injpar, \surjpar, \negpar)\text{-}\game{Fun}{n}{k}(\str{A}, \str{B})$ describes a partial homomorphism $f \in \mathcal{S}$ then the forth condition on $\mathcal{S}$ guarantees that if Duplicator plays $\phi_f : A \rightarrow B$ in this round then the pebbled position at the end of the round will be $f^{\prime} \in \mathcal{S}$. As $\mathcal{S} \subset \textbf{Part}^{k}_{\negpar}(\str{A},\str{B})$ we know that such a move will not result in Duplicator losing the game. So if $\emptyset \in \mathcal{S}$, Duplicator can use $\mathcal{S}$ to play indefinitely without losing.\\
  For the other direction, we note that the set of possible positions when playing the game $(\injpar, \surjpar, \negpar)\text{-}\game{Fun}{n}{k}(\str{A}, \str{B})$ according to some winning Duplicator strategy $\Phi$ will form a forth system $\mathcal{S}_\Phi$.
\end{proof}

  Following Hella, we define the \textit{canonical} forth system for a game as follows:

\begin{defi}
  The canonical forth system for $(\injpar, \surjpar, \negpar)\text{-}\game{Fun}{n}{k}(\str{A}, \str{B})$ is denoted $I^{n,k}_{\mathbf{x}}(\str{A}, \str{B})$ and is given by the intersection $\bigcap_m I^{n,k,m}_{\mathbf{x}}(\str{A}, \str{B})$, whose conjuncts are defined inductively as follows:
  \begin{enumerate}
  \item  $I^{n,k,0}_{\mathbf{x}}(\str{A}, \str{B}) \defeq \textbf{Part}^{k}_{\negpar}(\str{A},\str{B})$.
  \item $I^{n,k,m+1}_{\mathbf{x}}(\str{A}, \str{B})$ is the set of $\rho \in I^{n,k,m}_{\mathbf{x}}(\str{A}, \str{B})$ such that $\rho$ satisfies the $(\injpar, \surjpar)$-forth condition with respect to the set $I^{n,k,m}_{\mathbf{x}}(\str{A}, \str{B})$
  \end{enumerate}
\end{defi}

It is not difficult to see that for any forth system $\mathcal{S}$ for $\mathbf{x}\text{-}\game{Fun}{n}{k}(\str{A}, \str{B})$ we have $\mathcal{S} \subset I^{n,k}_{\mathbf{x}}(\str{A}, \str{B})$. This means that there is a winning strategy for Duplicator  in the game $\mathbf{x}\text{-}\game{Fun}{n}{k}(\str{A}, \str{B})$ if, and only if, $I^{n,k}_{\mathbf{x}}(\str{A}, \str{B})$ is not empty.

To complete the vocabulary needed to emulate Hella's proof in this setting we introduce the following generalisations of Hella's definitions.

\begin{defi}
  For any $\rho \in \textbf{Part}^{k}_{\negpar}(\str{A}, \str{B})$ and $\phi(\tup{y})$ a formula in some logic, we say that $\rho$ \emph{preserves the validity} of $\phi(\tup{y})$ if for any $\tup{a} \subset \text{dom}(\rho)$ of the same length as $\tup{y}$ we have that $\str{A}, \tup{a} \models \phi(\tup{y}) \implies \str{B}, \rho(\tup{a}) \models \phi(\tup{y})$.
  Denote by $J^{n,k}_{\mathbf{x}}(\str{A}, \str{B})$ the set of all $\rho \in \textbf{Part}^{k}_{\negpar}(\str{A}, \str{B})$ which preserve the validity of all $\eplnkinftyomega{n,k}_{\tup{x}}$ formulas.
  Let $\epFO^{n,k}_{\mathbf{x}}$ denote the fragment of $\eplnkinftyomega{n,k}_{\tup{x}}$ with only finitary conjunctions and disjunctions.
  Denote by $K^{n,k}_{\mathbf{x}}(\str{A}, \str{B})$ the set of all $\rho \in \textbf{Part}^{k}_{\negpar}(\str{A}, \str{B})$ which preserve the validity of all $\epFO^{n,k}_{\mathbf{x}}$ formulas.
\end{defi}

Now, we directly modify Hella's argument to prove the following:
\begin{lem}\label{lem:hella}
  For $\str{A}, \str{B}$ finite relational structures, \[I^{n,k}_{\mathbf{x}}(\str{A}, \str{B}) = J^{n,k}_{\mathbf{x}}(\str{A}, \str{B}) = K^{n,k}_{\mathbf{x}}(\str{A}, \str{B})\]
\end{lem}

\begin{proof}
  We prove the result by showing that
  \[
I^{n,k}_{\mathbf{x}}(\str{A}, \str{B}) \subset J^{n,k}_{\mathbf{x}}(\str{A}, \str{B}) \subset K^{n,k}_{\mathbf{x}}(\str{A}, \str{B}) \subset I^{n,k}_{\mathbf{x}}(\str{A}, \str{B})
  \]
  The inclusion $J^{n,k}_{\mathbf{x}}(\str{A}, \str{B}) \subset K^{n,k}_{\mathbf{x}}(\str{A}, \str{B})$ is obvious so we focus on proving
  \begin{enumerate}
    \item $I^{n,k}_{\mathbf{x}}(\str{A}, \str{B}) \subset J^{n,k}_{\mathbf{x}}(\str{A}, \str{B})$; and
    \item $K^{n,k}_{\mathbf{x}}(\str{A}, \str{B}) \subset I^{n,k}_{\mathbf{x}}(\str{A}, \str{B})$
  \end{enumerate}
\textit{Proof of 1.} Given $\rho \in I^{n,k}_{\mathbf{x}}(\str{A}, \str{B})$ we prove by structural induction on $\phi \in \eplnkinftyomega{n,k}_{\mathbf{x}}$ that $p$ preserves $\phi$. Clearly as $\rho$ is a partial homomorphism, it preserves atoms and, if $\negpar = 1$, negated atoms. The inductive cases for $\vee$ and $\wedge$ are easy so we focus on the cases where
\[
    \phi(\mathbf{z}) = \mathbf{Q} \mathbf{y} (\psi_1(\mathbf{y}_1, \mathbf{z}_1), \dots \psi_m(\mathbf{y}_m, \mathbf{z}_m))
\]
Now $\rho \in I^{n,k}_{\mathbf{x}}(\str{A}, \str{B})$ implies the existence of a map $f: A \rightarrow B$ such that for all $C \subset \text{dom}(\rho), D \subset A$ with $|D| \leq n$ we have $(\rho \downharpoonright C) \cup (f \downharpoonright D) \in I^{n,k}_{\mathbf{x}}(\str{A}, \str{B})$, so using the induction hypothesis we have that for all $i$, and tuples $\mathbf{a}_i$ from $D$ and $\mathbf{b}_i$ from $C$,
\[ \str{A}, \mathbf{a}_i, \mathbf{b}_i \models \psi_i(\mathbf{y}_i, \mathbf{z}_i) \implies \str{B}, f\mathbf{a}_i, \rho\mathbf{b}_i \models \psi_i(\mathbf{y}_i, \mathbf{z}_i) \]

This means that $f$ is a homomorphism
\[ f: \langle A, \psi_1(\cdot,\mathbf{b}_1), \dots \psi_m(\cdot,\mathbf{b}_m)\rangle  \rightarrow \langle B, \psi_1(\cdot,\rho\mathbf{b}_1), \dots \psi_m(\cdot,\rho\mathbf{b}_m)\rangle\]

Furthermore, in the cases where $(\injpar, \surjpar) = (1, 0), (0, 1)$ or $(1,1)$ this homomorphism is injective, surjective and bijective respectively and the quantifier $\mathbf{Q}$ in general represents a query $\mathcal{K}$ which is closed by injective-homomorphism, surjective-homomorphism or bijective-homomorphism so in all of these cases
\[ \langle A, \psi_1(\cdot,\mathbf{b}_1), \dots \psi_m(\cdot,\mathbf{b}_m)\rangle \in \mathcal{K} \implies  \langle B, \psi_1(\cdot, \rho\mathbf{b}_1), \dots \psi_m(\cdot, \rho\mathbf{b}_m)\rangle \in \mathcal{K} \]

and so $\str{A}, \mathbf{b} \models \phi(\mathbf{z}) \implies \str{B}, \rho\mathbf{b} \models \phi(\mathbf{z})$ and we are done with Part 1 of the proof.

\textit{Proof of 2.} Suppose that we have $p \in K^{n,k}_{\mathbf{x}}(\str{A}, \str{B})$. We have that $p \in I^{n,k,0}_{\mathbf{x}}(\str{A}, \str{B})$ by definition, so we prove by induction that $p \in I^{n,k, m}_{\mathbf{x}}(\str{A}, \str{B})$, for all $m$. Indeed, suppose this is true for $m^{\prime} < m$ but that $p \not\in I^{n,k, m}_{\mathbf{x}}(\str{A}, \str{B})$. Then it must be the case that for every $f: A \rightarrow B$ (injective if $\injpar = 1$, surjective if $\surjpar=1$) there is some choice of tuples $\mathbf{b}_f$ from $\text{dom}(p)$ and $\mathbf{a}_f$ from $A$ with $|\mathbf{a}_f| \leq n$ and $|\mathbf{a}_f \cup \mathbf{b}_f| \leq k$  such that $f \downharpoonright \mathbf{a}_f) \cup (p \downharpoonright \mathbf{b}_f) \not\in I^{n,k, m-1}_{\mathbf{x}}(\str{A}, \str{B})$.  By induction, this means that $(f \downharpoonright \mathbf{a}_f) \cup (p \downharpoonright \mathbf{b}_f) \not\in  K^{n,k}_{\mathbf{x}}(\str{A}, \str{B})$ and so there is a  formula $\psi_f(\mathbf{y}, \mathbf{z})$ such that $\str{A}, \mathbf{a}_f, \mathbf{b}_f \models \psi_f(\mathbf{y}, \mathbf{z}) $ but $  \str{B}, p\mathbf{a}_f, f\mathbf{b}_f \not\models \psi_f(\mathbf{y}, \mathbf{z}).$

Let $F_{\mathbf{x}}$ denote the set of functions $f: A \rightarrow B$ which are injective if $\injpar=1$ and surjective if $\surjpar =1$. Recall from Observation~\ref{obs:size}, the existence of $p$ implies that $F_{\mathbf{x}}$ is non-empty.  Now we define two structures $\str{A}_p = \langle A, (\psi_f(\cdot,\mathbf{b}_f))_{F_{\mathbf{x}}} \rangle$ and $\str{B}_p = \langle B, (\psi_f(\cdot, p\mathbf{b}_f))_{F_{\mathbf{x}}} \rangle$. We have by construction that no $f \in F_{\mathbf{x}}$ is a homomorphism from $\str{A}_p \rightarrow \str{B}_p$, meaning that we can define a query $E$ with $\str{A}_p \in E$ and $\str{B}_p \notin E$ which is closed under:
\begin{itemize}
  \item all homomorphisms, if $(\injpar, \surjpar) = (0, 0)$
  \item all injective homomorphisms, if $(\injpar, \surjpar) = (1, 0)$
  \item all surjective homomorphisms, if $(\injpar, \surjpar) = (0, 1)$
  \item all bijective homomorphisms, if $(\injpar, \surjpar) = (1, 1)$
\end{itemize}
So in all cases, the quantifier $Q_E$ is allowed in $\eplnkinftyomega{n,k}_{\mathbf{x}}$.

Since each formula  $\psi_f$ is in $\epFO^{n,k}_{\mathbf{x}}$, it has at most $k$ free variables in all.  By renaming these variables, we can ensure that the variables $\mathbf{y}_f$ are all from among a fixed tuple of $n$ variables $\mathbf{y}$ which are distinct from all variables in all $\mathbf{z}_f$.  Then
\[\phi(\mathbf{y}) = \mathcal{Q}_E \mathbf{y}. (\psi_f(\mathbf{y}_f, \mathbf{z}_f))_{f \in F_{\mathbf{x}}} \]
is a formula of $\epFO^{n,k}_{\mathbf{x}}$, since each sub-formula still has at most $k$ free variables (recall Remark~\ref{rem:typeAtypeB}).  This formula is true on $(\str{A}_p, \mathbf{b})$ but false on $(\str{B}_p, p\mathbf{b})$. However, this contradicts that $p \in  K^{n,k}_{\mathbf{x}}(\str{A}, \str{B})$ and so preserves the truth of all such formulas.
\end{proof}

We conclude this section by showing the desired correspondence for the whole family of games and logics we have introduced.

\begin{thm}\label{thm:logic}
  For $\mathbf{x} \in \{0,1\}^{3}$ and all $n, k \in \nats$ the following are equivalent:
  \begin{itemize}
    \item Duplicator has a winning strategy for $\mathbf{x}$-$\game{Fun}{n}{k}(\str{A}, \str{B})$
    \item $\str{A} \Rrightarrow_{\eplnkinftyomega{n,k}_{\mathbf{x}}} \str{B}$
    \item $\str{A} \Rrightarrow_{\epFO^{n,k}_{\mathbf{x}}} \str{B}$
  \end{itemize}
\end{thm}
\begin{proof}
 First note that by the definition of the canonical forth system, Duplicator wins $\mathbf{x}$-$\game{Fun}{n}{k}(\str{A}, \str{B})$ if, and only if, $\emptyset \in I^{n,k}_{\mathbf{x}}(\str{A}, \str{B})$.\\
 Furthermore, $J^{n,k}_{\mathbf{x}}(\str{A}, \str{B})$ and $K^{n,k}_{\mathbf{x}}(\str{A}, \str{B})$ are defined as the sets of partial maps $\rho$ which preserve any $\eplnkinftyomega{n,k}_\mathbf{x}$ or $\epFO^{n,k}_{\mathbf{x}}$ formulas respectively  which hold on the domain of $\rho$.  So $\emptyset \in J^{n,k}_{\mathbf{x}}(\str{A}, \str{B})$ or $K^{n,k}_{\mathbf{x}}(\str{A}, \str{B})$ if, and only if, all \textit{sentences} in these logics which are true $\str{A}$ are also true in $\str{B}$, i.e.\ $\str{A} \Rrightarrow_{\eplnkinftyomega{n,k}_{\mathbf{x}}} \str{B}$ or $\str{A} \Rrightarrow_{\epFO^{n,k}_{\mathbf{x}}} \str{B}$.
Applying the result of Lemma~\ref{lem:hella} proves the equivalence of these three.
\end{proof}

\section{The Hella Comonad and its Kleisli Category}\label{sec:comonad}

In this section, we show how to construct a game comonad $\G{n,k}$ which captures the strategies of $+\game{Fun}{n}{k}$ in the same way that $\T{k}$ captures the strategies of $\exists\game{Peb}{k}{}$. We do this using a new technique for constructing new game comonads from old based on strategy translation. We then show that different types of morphism in the Kleisli category of this new comonad correspond to Duplicator strategies for the games introduced in Section~\ref{sec:logic}.

\subsection{Translating between games} \label{sec:translate}

The pebbling comonad is obtained by defining a structure $\T{k}\str{A}$ for each $\str{A}$ whose universe consists of (non-empty) lists in $(A \times [k])^{\ast}$ which we think of as sequences of moves by Spoiler in a game $\exists\game{Peb}{k}{}(\str{A}, \str{B})$, with $\str{B}$ unspecified.  With this in mind, we call a sequence in $(A \times [k])^{\ast}$ a \emph{$k$-history} (allowing the empty sequence).  In contrast, a move in the +$\game{Fun}{n}{k}(\str{A}, \str{B})$ involves Spoiler moving up to $n$ pebbles and therefore a history of Spoiler moves is a sequence in $((A \times [k])^{\leq n})^{\ast}$.  We call such a sequence an \emph{$n,k$-history}.
With this set-up, (deterministic) strategies are given by functions \[((A \times [k])^{\ast}\times [k]) \rightarrow (A \rightarrow B) \] for $\exists\game{Peb}{k}{}(\str{A}, \str{B})$ and \[ ((A \times [k])^{\leq n})^{\ast} \rightarrow (A \rightarrow B)\] for +$\game{Fun}{n}{k}(\str{A}, \str{B})$.

A winning strategy for Duplicator  in +$\game{Fun}{n}{k}(\str{A}, \str{B})$ can always be translated into one in  $\exists\game{Peb}{k}{}(\str{A}, \str{B})$.  We aim now to establish conditions for when a translation can be made in the reverse direction.  For this, it is useful to establish some machinery.

There is a natural \emph{flattening} operation that takes $n,k$-histories to $k$-histories.  We denote the operation by $F$, so $F([s_1,s_2,\ldots,s_m]) = s_1\cdot s_2 \cdots s_m$, where $s_1,\ldots, s_m \in (A \times [k])^{\leq n}$.  Of course, the function $F$ is not injective and has no inverse.  It is worth, however, considering functions $G$ from $k$-histories to $n,k$-histories that are inverse to $F$ in the sense that $F(G(t)) = \nolinebreak t$.  One obvious such function takes a $k$-history $s_1,\ldots,s_m$ to the $n,k$-history $[[s_1],\ldots,[s_m]]$, i.e.\ the sequence of one-element sequences.  This is, in some sense, minimal in that it imposes the minimal amount of structure on $G(t)$.  We are interested in a \emph{maximal} such function.  For this, recall that the sequences in $(A\times[k])^{\leq n}$ that form the elements of an $n,k$-history have length at most $n$ and do not have a repeated index from $[k]$.  We aim to break a $k$-history $t$ into maximal such blocks.  This leads us to the following definition.
\begin{defi}
  A list $s \in (A \times [k])^{\ast}$ is called \emph{basic} if it contains fewer than or equal to $n$ pairs and the pebble indices are all distinct.

  The $n$-structure function $S_n: (A \times [k])^{\ast} \rightarrow ((A \times [k])^{\leq n})^{\ast}$ is defined recursively as follows:
 \begin{itemize}
    \item $S_n(s) = [s]$ if $s$ is basic
    \item otherwise, $S_n(s) = [a] ; S_n(t)$ where $s = a\cdot t$ such that $a$ is the largest basic prefix of $s$.
  \end{itemize}
\end{defi}
It is immediate from the definition that $F(S_n(t)) = t$.
It is useful to characterise the range of the function $S_n$, which we do through the following definition.
\begin{defi}\label{def:nk_structured}
  An $n,k$-history $t$ is \emph{structured} if whenever $s$ and $s'$ are successive elements of $t$, then either $s$ has length exactly $n$ or $s'$ begins with a pair $(a,p)$ such that $p$ occurs in $s$.
\end{defi}

It is immediate from the definitions that $S_n(s)$ is structured for all $k$-histories $s$ and that an $n,k$-history is structured if, and only if, $S_n(F(s)) = s$.

We are now ready to characterise those Duplicator winning strategies for $\exists\game{Peb}{k}{}$  that can be lifted to $+\game{Fun}{n}{k}$.  First, we define a function that lifts a position in $\exists\game{Peb}{k}{}$  that Duplicator must respond to, i.e.\ a pair $(s,p)$ where $s$ is a $k$-history and $p$ a pebble index, to a position in  $+\game{Fun}{n}{k}$, i.e.\ an $n,k$-history.

\begin{defi}
  Suppose $s$ is a $k$-history and $s'$ is the last basic list in $S_n(s)$, so $S_n(s) = t;[s']$.  Let $p\in[k]$ be a pebble index.

Define the \emph{$n$-structuring}  $\alpha_n(s,p)$ of $(s,p)$ by
\[
  \alpha_n(s,p) = \begin{cases}
   t;[s'] \text{ if  } |s'|=n \text{ or } p \text{ occurs in } s'\\
    t\phantom{;[s']} \text{ otherwise.}
  \end{cases}
\]
\end{defi}

\begin{defi}
  Say that a Duplicator strategy $\Psi: ((A \times [k])^{\ast}\times [k]) \rightarrow (A \rightarrow B)$ in $\exists\game{Peb}{k}{}$ is $n$-consistent if for all $k$-histories $s$ and $s'$ and all pebble indices $p$ and $p'$:
  \[ \alpha_n(s,p) = \alpha_n(s',p') \quad \Rightarrow \quad \Psi(s,p) = \Psi(s',p'). \]
\end{defi}

Intuitively, an $n$-consistent Duplicator strategy in the game $\exists\game{Peb}{k}{}(\str{A}, \str{B})$ is one where Duplicator plays the same function in all moves that could be part of the same Spoiler move in the game +$\game{Fun}{n}{k}(\str{A}, \str{B})$.  We are then ready to prove the main result of this subsection.

\begin{lem}\label{lem:lifting-strategy}
  Duplicator has an $n$-consistent winning strategy in $\exists\game{Peb}{k}{}(\str{A}, \str{B})$ if, and only if, it has a winning strategy in +$\game{Fun}{n}{k}(\str{A}, \str{B})$.
\end{lem}
\begin{proof}
  The reverse direction is easy. Suppose first that $\Psi: ((A \times [k])^{\leq n})^{\ast} \rightarrow (A \rightarrow B)$ is a Duplicator winning strategy in +$\game{Fun}{n}{k}(\str{A}, \str{B})$.  Define the strategy $\Psi'$ in $\exists\game{Peb}{k}{}(\str{A}, \str{B})$ such that for a $k$-history $s$ and a pebble index $p\in [k]$, $\Psi'(s,p) = \Psi(\alpha_n(s,p))$.  This is easily seen to be $n$-consistent and winning.

  For the other direction we deal with the case of $n=1$ separately.

  For $n=1$, all $1,k$-histories are structured. Indeed, for any
  $k$-history $s$ and any pebble index $p$, $\alpha_1(s,p) = G(s)$. This means that
  for any $p$ and $p'$ $\alpha_1(s,p) = \alpha_1(s, p')$  and the $1$-consistent winning strategies are precisely those such that for any
  $k$-history $s$, pebble indices $p$ and $p'$ and elements $a, b$ if $a = b$ then $f([s;(a,p)]) = f([s;(b,p')])$. This is the same as saying that the branch maps $\phi^f_{s,p}$ and $\phi^f_{s,p'}$ are equal for every history $s$ and every pair of pebble indices $p$ and $p'$. We denote the common branch map at $s$ by $\phi^f_{s}$. This then gives a strategy in the game +$\game{Fun}{1}{k}(\str{A}, \str{B})$ where after every $1,k$-history $t$, Duplicator provides the function $\phi^f_{F(t)}$.

For $n \geq 2$, suppose $\Psi$ is an $n$-consistent winning strategy for Duplicator in $\exists\game{Peb}{k}{}(\str{A}, \str{B})$.  We construct from this a winning strategy $\Psi'$ for Duplicator in +$\game{Fun}{n}{k}(\str{A}, \str{B})$.  If $t$ is a structured $n,k$-history and $p$ is the last pebble index occurring in it, we can just take $\Psi'(t) = \Psi(F(t),p)$.  To extend this to unstructured $n,k$-histories, we first define the structured companion of an $n,k$-history.

  Suppose $t$ is an $n,k$-history that is not structured and let $s,s' \in (A\times [k])^{\leq n}$ be a pair of consecutive sequences witnessing this.  We call such a pair a \emph{bad} pair.  Let $(a,p)$ be the last pair occurring in $s$ and $(a',p')$ the first pair occurring in $s'$.
  Let $t'$ be the prefix of $t$ ending with $s$ and let $\kappa$ be the last element of $A$ such that $(\kappa,p')$ appears in $F(t)$ if there is any.  We now obtain a new $n,k$-history from $t$ by replacing the pair $s,s'$ by $s,\texttt{link},s'$ where
\[
\texttt{link} = \begin{cases}
                    [(a, p), (\kappa,p')] \quad \text{if defined}\\
                    [(a,p), (a,p')]  \quad \text{otherwise}.
                    \end{cases}
                  \]
                  It is clear that in this $n,k$-history, neither of the pairs $s,\texttt{link}$ or $\texttt{link},s'$ is bad, so it has one fewer bad pair than $t$. Also, this move is chosen so that responding to the moves $F(\texttt{link})$ according to $\Psi$ does not change the partial function defined by the pebbled position after responding to the moves $F(s)$. Repeating the process, we obtain a structured $n,k$-history which we call $\tilde{t}$, the structured companion of $t$.

 We can now formally define the Duplicator strategy by saying for any $n,k$-history $t$, $\Psi'(t) = \Psi(F(\tilde{t}),p)$ where $\tilde{t}$ is the structured companion of $t$ and $p$ is the last pebble index occurring in $t$. To see why $\Psi'$ is a winning strategy, we note that as responding with $\Psi$ to the $\texttt{link}$ moves does not alter the partial function defined by the pebbled position, the function defined after responding to $\tilde{t}$ according to $\Psi$ is the same as that defined after responding to $t$ according to $\Psi'$. So if there is a winning $n,k$-history $t$ for Spoiler against $\Psi'$ then $F(\tilde{t})$ is a winning $k$-history for Spoiler against $\Psi$, a contradiction.
\end{proof}

\subsection{\texorpdfstring{Lifting the comonad $\T{k}$ to $\G{n,k}$}{Lifting the comonad P\textsubscript{k} to H\textsubscript{n,k}}}

Central to Abramsky et al.'s construction of the pebbling comonad is the observation that for $I$-structures (defined in Section~\ref{sec:background}), maps in the Kleisli category $\mathcal{K}(\T{k})$ correspond to Duplicator winning strategies in $\exists\game{Peb}{k}{}(\str{A}, \str{B})$.

\begin{lem}[\cite{Abramsky2017}]\label{lem:abramsky}
For $\str{A}$ and $\str{B}$ $I$-structures over the signature $\sigma$, there is a homomorphism $\T{k}\str{A} \rightarrow \str{B}$ if, and only if, there is a (deterministic) winning strategy for Duplicator in the game $\exists\game{Peb}{k}{}(\str{A}, \str{B})$
\end{lem}

The relation to strategies is clear in the context of elements $s \in \T{k}\str{A}$ representing histories of Spoiler moves up to and including the current move in the $\exists\game{Peb}{k}{}(\str{A},\str{B})$. The relational structure given to this set by Abramsky, Dawar and Wang ensures that pebbled positions preserve relations in $\sigma$, while the caveat here about $I$-structures is a technicality to ensure that the pebbled positions when ``playing'' according to a map $f$ all define partial homomorphisms, in particular they give well defined partial maps from $A$ to $B$.

As we saw in Lemma~\ref{lem:lifting-strategy} a Duplicator winning strategy in +$\game{Fun}{n}{k}(\str{A}, \str{B})$ is given by an $n$-consistent strategy in $\exists\game{Peb}{k}{}(\str{A}, \str{B})$.  The $n$-consistency condition can be seen as saying that the corresponding map $f: \T{k}\str{A} \ra \str{B}$ must, on certain ``equivalent'' elements of $\T{k}\str{A}$ give the same value.  We can formally define the equivalence relation as follows.

\begin{defi}
  For $n \in \mathbb{N}$ and $\mathcal{A}$ a relational structure. Define $\approx_n$ on the universe of $\T{k}\str{A}$ as follows:
  \[ [s;(a,i)]\approx_{n} [t;(b,j)] \iff a = b \textbf{ and } \alpha_n((s,i)) = \alpha_n((t,j)) \]
\end{defi}
In general, for any structured $n,k$-history $t$, we write $[t | a]$ to denote the $\approx_n$-equivalence class of an element $[s;(a,i)] \in \T{k}\str{A}$ with $\alpha_n(s,i) = t$.

This allows us to define the main construction of this section as a quotient of the relational structure $\T{k}\str{A}$.  Note that the relation $\approx_n$ is not a congruence of this structure, so there is not a canonical quotient.  Indeed, given an arbitrary equivalence relation $\sim$ over a relational structure $\str{M}$, there are two standard ways to define relations in a quotient $\str{M}/{\sim}$.  We could say that a tuple $(c_1, \dots c_r)$ of equivalence classes is in a relation $R^{\str{M}/{\sim}}$ if, and only if, \emph{every} choice of representatives is in $R^{\str{M}}$ or if \emph{some} choice of representatives is in $R^{\str{M}}$.  The latter definition has the advantage that the quotient map from $\str{M}$ to $\str{M}/{\sim}$ is a homomorphism and it is this definition that we assume for the rest of the paper. From this definition we also see that for any homomorphism $f: \str{A} \rightarrow \str{B}$ the map $\T{k}f/\approx_n: \T{k}\str{A}/\approx_n \rightarrow \T{k}\str{B}/\approx_n$ which sends $[[s; (a,i)]] = [\alpha_n(s,i) | a]$ to $[\alpha_n(\T{k}f (s), i) | f(a)]$ is a well-defined homomorphism.

\begin{defi}
  For $n, k \in \mathbb{N}$, $k \geq n$ and $\sigma$ a relational signature, we define the functor $\G{n,k}: \mathcal{R}(\sigma) \rightarrow \mathcal{R}(\sigma)$ by:
  \begin{itemize}
    \item \textbf{On objects} $\G{n,k}\str{A} \defeq \T{k}\str{A}/{\approx_n}$.
    \item \textbf{On morphisms} $\G{n,k}f \defeq \T{k}f/{\approx_n}$.
  \end{itemize}
\end{defi}

Writing $q_n:  \T{k}\str{A} \rightarrow \G{n,k}\str{A}$ for the quotient map enables us to establish the following useful property.
\begin{obs} \label{obs:hom_fact}
\[
f: \G{n,k}\str{A} \rightarrow \str{B} \text{ is a homomorphism} \iff f \circ q_n: \T{k}\str{A} \rightarrow \str{B} \text{ is a homomorphism}
\]
\end{obs}

Combining this with Lemma~\ref{lem:lifting-strategy}, we have the appropriate generalisation of Lemma~\ref{lem:abramsky}.

\begin{lem}
  For $I$-structures $\str{A}$ and $\str{B}$, there is a homomorphism $f: \G{n,k}\str{A} \rightarrow \str{B}$ if, and only if, there is a winning strategy for Duplicator in the game +$\game{Fun}{n}{k}$
\end{lem}

\begin{proof}
  From right to left,  by  Lemma~\ref{lem:lifting-strategy} we have an $n$-consistent winning strategy $\Psi$ for Duplicator in  $\exists\game{Peb}{k}{}(\mathcal{A},\mathcal{B})$.  The $n$-consistency condition implies that the Duplicator  response to a Spoiler play $[s;(a,i)] \in (A\times [k])^{\ast}$ is determined by $\alpha_n((s,i))$ and $a$ only.  So the corresponding homomorphism $f_{\Psi}: \T{k}A \rightarrow B$ respects $\approx_n$ and $f_{\Psi} \circ q_n$ is a well-defined homomorphism $f: \G{n,k}A \rightarrow B$ .

  For the other direction, note that $f\circ q_n$ defines a Duplicator winning strategy for $\exists\game{Peb}{k}{}(\str{A}, \str{B})$ which is $n$-consistent.   Thus, by  Lemma~\ref{lem:lifting-strategy},  there is a winning strategy for Duplicator in +$\game{Fun}{n}{k}(\str{A}, \str{B})$.
\end{proof}

Furthermore, we can see that the quotient map $q_n$ defined above is indeed a
natural transformation between the functors $\T{k}$ and $\G{n,k}$.

\begin{lem}\label{lem:qn_nat}
  $q_n : \T{k} \Rightarrow \G{n,k}$ is a natural transformation.
\end{lem}

\begin{proof}

Let $\str{A}$ and $\str{B}$ be relational structures over the same signature and $f: \str{A} \rightarrow \str{B}$
be a homomorphism. To show that $q_n$ is natural we need to establish the equality
$q_n \circ \T{k} f = \G{n,k} f \circ q_n$. Fix an element $[s; (a,i)] \in \T{k} \str{A}$.
On the right hand side of we have that $q_n([s; (a,i)]) = [\alpha_n(s,i) | a]$ and so
$\G{n,k} f \circ q_n = [\alpha_n(\T{k} f (s),i) | f(a)]$. On the left hand side,
$\T{k} f ([s;(a,i)]) = [\T{k}f (s); (f(a), i)]$ and so $q_n \circ \T{k} f ([s;(a,i)])
= [\alpha_n(\T{k} f (s),i) | f(a)]$ as required.
\end{proof}

This allows us to prove the following important lemma.

\begin{lem}
The counit $\epsilon$ and comultiplication $\delta$ for $\T{k}$ lift to well-defined natural transformations for $\G{n,k}$.

\end{lem}

\begin{proof}
  Suppose $[s;(a,i)] \approx_n [t;(b,j)] \in \T{k}\str{A}$. Then by the definition of $\approx_n$, we have $a=b$ and so $\epsilon_{\str{A}}([s;(a,i)] )=a=b=\epsilon_{\str{A}}([t;(b,j)])$ so we can define $\epsilon^{n,k}_{\str{A}}: \G{n,k}\str{A} \rightarrow \str{A}$ such that $\epsilon_{\str{A}} = \epsilon^{n,k}_{\str{A}} \circ q_n$. So by Observation~\ref{obs:hom_fact} this is a homomorphism for every $\str{A}$ and by Lemma~\ref{lem:qn_nat} it is natural. \\
  The argument is slightly more complicated for $\delta$. Firstly introduce $\delta^{\prime}: ((A \times [k])^{\leq n})^{\ast} \rightarrow ((\T{k}A \times [k])^{\leq n})^{\ast}$ defined such that the length of the $i^{\text{th}}$ list in $\delta^{\prime}(s)$ is the same as the length of the $i^{\text{th}}$ list in $s$ and $F(\delta^{\prime}(s)) = \delta(F(s))$. Informally, this means replacing every $a \in A$ appearing in $s$ with the prefix of $F(s)$ which runs up to (and includes) that appearance of $a$. Now it is not hard to see that for any $s \in \T{k}\str{A}$
  \[
  \alpha_n(\delta(s), i) = \delta^{\prime}(\alpha_n(s, i))
  \]
  Now, as $\delta$ is a map from $\T{k}\str{A}$ to $\T{k}\T{k}\str{A}$, to show that it ``lifts'' to being a comultiplication for $\G{n,k}$ we must show that the function \[\T{k}\str{A} \xrightarrow{\delta} \T{k}\T{k}\str{A} \xrightarrow{\T{k}q_n} \T{k}\G{n,k}\str{A} \xrightarrow{q_n} \G{n,k}\G{n,k}\str{A}\] is well-defined with respect to $\approx_n$.
  So, for any $[s;(a,i)] \approx_n [t;(b,j)]$, we prove that $\T{k}q_n \circ\delta([s;(a,i)]) \approx_n \T{k}q_n \circ\delta([t;(b,j)])$, as elements of $\T{k}\G{n,k}\str{A}$. Firstly, by definition $\delta([s;(a,i)]) = [\delta(s); ([s;(a,i)], i)]$ and so $\T{k}q_n \circ\delta([s;(a,i)]) =  [\T{k}q_n(\delta(s)); (q_n([s;(a,i)]), i)]$. We can write similar expressions for $[t;(b,j)]$. \\
  As we have that $\alpha_n(s,i) = \alpha_n(t,j)$ we use the above fact about $\delta^{\prime}$ to get that $\alpha_n(\delta(s), i) = \alpha_n(\delta(t), i)$. As $\T{k}q_n$ only changes the elements of a list leaving the pebble indices unchanged and $\alpha_n$ is based only on the pebble indices of a list, we can deduce that $\alpha_n(\T{k}q_n(\delta(s)), i) = \alpha_n(\T{k}q_n(\delta(t)), i)$. So, by the definition of $\approx_n$, $\T{k}q_n \circ\delta([s;(a,i)]) \approx_n \T{k}q_n \circ\delta([t;(a,i)])$ if $q_n([s;(a,i)]) = q_n([t;(b,j)])$, which is precisely the statement that  $[s;(a,i)] \approx_n [t;(b,j)]$.\\
  Naturality for $\delta$ follows from the naturality of $q_n$ and the naturality
  of the comultiplication of $\T{k}$.
\end{proof}

We call these lifted natural transformations $\epsilon^{n,k} : \G{n,k} \rightarrow 1$ and $\delta^{n,k}: \G{n,k} \rightarrow \G{n,k}\G{n,k}$. As $q_n \circ \T{k}q_n = \G{n,k}q_n \circ q_n$, we have that for any $t \in (\T{k})^{m}\str{A}$ the notion of ``the'' equivalence class of $t$, $\mathbf{q_n}(t) \in (\G{n,k})^{m}\str{A}$ is well-defined. So for any term $T$ built from composing $\epsilon, \delta$ and $\T{k}$ we have that the term $\tilde{T}$, obtained by replacing $\epsilon$ by $\epsilon^{n,k}$, $\delta$ with $\delta^{n,k}$ and $\T{k}$ with $\G{n,k}$ satisfies $\mathbf{q_n}(T(t)) = \tilde{T}(\mathbf{q_n}(t))$ by the above proof. Now as the counit and coassociativity laws are equations in $\epsilon$ and $\delta$ which remain true on taking the quotient we have the following result.

\begin{thm}
  $(\G{n,k}, \epsilon^{n,k}, \delta^{n,k})$ is a comonad on $\mathcal{R}(\sigma)$
\end{thm}

\subsection{\texorpdfstring{Classifying the morphisms of $\mathcal{K}(\G{n,k})$}{Classifying the morphisms of K(H\textsubscript{n,k})}}

In Abramsky et al.'s treatment of the Kleisli category of $\T{k}$~\cite{Abramsky2017} they classify the morphisms according to whether their \emph{branch maps} are injective, surjective or bijective.  We extend this definition to the comonad $\G{n,k}$.  This gives us a way of classifying the morphisms to match the classification of strategies given in Section~\ref{sec:logic}.

\begin{defi}
 For $f : \G{n,k}\str{A} \rightarrow \str{B}$ a Kleisli morphism of $\G{n,k}$, the \emph{branch maps} of $f$ are defined as the following collection of  functions $A \rightarrow B$, indexed by the structured $n,k$-histories $t \in ((A \times [k])^{\leq n})^{\ast}$:
 \[\phi^{f}_t(x) = f([t | x]).\]
  We say that such an $f$ is \emph{branch-bijective} (resp.\ \emph{branch-injective, -surjective}) if for every $t$
   \begin{align*}
     \phi^{f}_{t} &\text{is bijective (resp.\ injective, surjective)}
  \end{align*}
  We denote these maps by $\str{A} \bbij{n,k} (resp.\ \str{B}$ $\str{A} \binj{k} \str{B}$ and $\str{A} \bsurj{k} \str{B}$)
\end{defi}

Informally, the branch map $\phi^{g}_s$ is the response given by Duplicator in the +$\game{Fun}{n}{k}(\str{A}, \str{B})$ when playing according to the strategy represented by $g$ after Spoiler has made the series of plays in $s$.  This gives us another way of classifying the Duplicator winning strategies for the games from Section~\ref{sec:logic}.
\begin{lem}
  There is a winning strategy for Duplicator in the game +$\game{Bij}{n}{k}(\str{A}, \str{B})$ (resp.\ +$\game{Inj}{n}{k}(\str{A}, \str{B})$, +$\game{Surj}{n}{k}(\str{A}, \str{B})$) if and only if $\str{A} \bbij{n,k} \str{B}$ (resp.\ $\str{A} \binj{n,k} \str{B}$, $\str{A} \bsurj{n,k} \str{B}$).
\end{lem}

\begin{proof}
  Immediate from the definitions.
\end{proof}

Expanding this connection between Kleisli maps and strategies, we define the following:
\begin{defi}
  We say a a Kleisli map $f :\G{n,k}\str{A} \rightarrow \str{B}$ is \emph{strongly} branch-bijective (resp.\ \emph{strongly} branch-injective, -surjective) if the strategy for the game $+\game{Bij}{n}{k}(\str{A}, \str{B})$ (resp.\ $+\game{Inj}{n}{k}(\str{A}, \str{B}), +\game{Surj}{n}{k}(\str{A}, \str{B})$) is also a winning strategy for the game $\game{Bij}{n}{k}(\str{A}, \str{B})$ (resp.\ $\game{Inj}{n}{k}(\str{A}, \str{B}), \game{Surj}{n}{k}(\str{A}, \str{B})$) and we denote these maps by $\str{A} \sbbij{n,k} \str{B}$ (resp.\ $\str{A} \sbinj{k} \str{B}$ and $\str{A} \sbsurj{k} \str{B}$)
\end{defi}

Now we generalise a result of Abramsky, Dawar and Wang to the Kleisli category $\mathcal{K}(\G{n,k})$.

\begin{lem}\label{lem:kleisli_equiv}
  For $\str{A}, \str{B}$ finite relational structures, \[\str{A} \bfinj{n,k} \str{B} \iff \str{A} \bfsurj{n,k} \str{B} \iff \str{A} \sbbij{n,k} \str{B} \iff \str{A} \cong_{\mathcal{K}(\G{n,k})} \str{B}\]
\end{lem}

\begin{proof}
  As $A$ and $B$ are finite, the existence of an injection $A \rightarrow B$ implies that $\abs{A} \leq \abs{B}$. So, $\str{A} \bfinj{n,k} \str{B}$ implies that $\abs{A} = \abs{B}$ and thus any injective map between the two is also surjective and vice versa. This means the first equivalence is trivial and further both of these imply $\str{A} \bfbij{n,k} \str{B}$ \\
  For the second equivalence, we first introduce some notation. Let $P^{m}_{\str{A}}$ be the finite substructure of $\G{n,k}\str{A}$ induced on the elements $\{ [s | a] \mid s \in ((A \times [k])^{\leq n})^{\leq m} \}$. Note that for any $f: \str{A} \bbij{n,k} \str{B}$, the Kleisli completion $f^{\ast}$ restricts to a bijective homomorphism $P^{m}_{\str{A}} \rightarrow P^{m}_{\str{B}}$ for each $m$. So if $f: \G{n,k}\str{A}\rightarrow \str{B}$ and $g: \G{n,k}\str{B}\rightarrow \str{A}$ are branch-bijective, we have for each $m$ a pair of bijective homomorphisms $P^{m}_{\str{A}} \rightleftarrows P^{m}_{\str{B}}$. As these are finite structures we can deduce that these are indeed isomorphisms and so $f$ is a strategy for $\game{Bij}{n}{k}(\str{A}, \str{B})$.\\
  For the final equivalence, if $f$ witnesses $\str{A}\sbbij{n,k} \str{B}$ then we have, by induction, that $f^{\ast}$ is an isomorphism from $P^{m}_{\str{A}}$ to $P^{m}_{\str{B}}$ for each $m$. So $f^{\ast}: \G{n,k}\str{A} \rightarrow \G{n,k}\str{B}$ is an isomorphism witnessing $\str{A}\cong_{\mathcal{K}(\G{n,k})} \str{B}$. For the converse we suppose that there is an isomorphism $h^{\ast}: \G{n,k}\str{A} \rightarrow \G{n,k}\str{A}$. Then the Kleisli map $h = \epsilon_{\str{B}} \circ h^{\ast}$ is a strongly branch-bijective strategy.
\end{proof}

This lemma allows us to conclude that the isomorphisms in the category $\mathcal{K}(\G{n,k})$ correspond with equivalence of structures up to $k$ variable infinitary logic extended by all generalised quantifiers of arity at most $n$ and thus with winning strategies for Hella's $n$-bijective $k$-pebble game.

\begin{thm}
  For two $I$-structures $\str{A}$ and $\str{B}$ the following are equivalent:
  \begin{itemize}
    \item $\str{A} \cong_{\mathcal{K}(\G{n,k})} \str{B}$
    \item Duplicator has a winning strategy for $\game{Bij}{n}{k}(\str{A},\str{B})$
    \item $\str{A} \equiv_{\qfk{k}(\mathcal{Q}_n)} \str{B}$
  \end{itemize}
\end{thm}
\begin{proof}
  Immediate from Lemma~\ref{lem:kleisli_equiv} and Hella~\cite{Hella1989}.
\end{proof}

A similar result can also be obtained relating branch-injective and branch-surjective
maps to monomorphisms and epimorphisms respectively. However, the category in question is not
the full category $\mathcal{K}(\G{n,k})$ where $\G{n,k}$ is seen as a comonad on
$\strs{\sigma^+}$ but rather the restriction of this category where the objects only the $I$-structures.
Abramsky and Shah show that this category can be obtained from the relevant game comonad
as the Kleisli category of a relative comonad~\cite{Abramsky2018}.

\section{Coalgebras and Decompositions}\label{sec:coalgebra}

Abramsky et al.~\cite{Abramsky2017} show that the coalgebras of the comonad $\T{k}$ have a surprising correspondence with objects of great interest to finite model theorists.  That is, any coalgebra $\alpha: \str{A} \rightarrow \T{k}\str{A}$ gives a tree decomposition of $\str{A}$ of width at most $k-1$ and  any such tree decomposition can be turned into a coalgebra.
This result works because $\T{k}\str{A}$ has a treelike structure where any pebble history, or branch,  $s \in \T{k}\str{A}$ only witnesses the relations from the $\leq k$ elements of $\str{A}$ which make up the pebbled position on $s$. So a homomorphism $\str{A} \rightarrow \T{k}\str{A}$ witnesses a sort of treelike $k$-locality of the relational structure $\str{A}$ and the $\T{k}$-coalgebra laws are precisely enough to ensure this can be presented as a tree decomposition (of width $<k$).

In lifting this comonad to $\G{n,k}$, we have given away some of the restrictive $k$-local nature of $\T{k}$ which makes this argument work. The structure $\G{n,k}\str{A}$ witnesses many more of $\str{A}$'s relations than $\T{k}\str{A}$. Take, for example, the substructure induced on the elements $\{[\epsilon | x ] \ | \ x \in \str{A}\}$, where $\epsilon$ is the empty history.  This witnesses all relations in $\str{A}$ which have arity $\leq n$. So, in particular, if $\str{A}$ contains no relations of arity greater than $n$, this substructure is just a copy of $\str{A}$ and the obvious embedding $A \rightarrow \G{n,k}A$ can be easily seen to be a $\G{n,k}$-coalgebra.
From this, we can see that if $\G{n,k}$-coalgebras capture some notion of $n$-generalised tree decomposition, this should clearly be more permissive than the notion of tree decomposition,  allowing a controlled amount of non-locality (parameterised by $n$) and collapsing completely for $\sigma$-structures with $n \geq \text{arity}(\sigma)$. In this section we define the appropriate generalisation of tree decomposition and show its relation with $\G{n,k}$-coalgebras.

\subsection{Generalising tree decomposition}

Recall the  definition of a tree decomposition of a $\sigma$-structure, for
example from Definition 4.1.1 of~\cite{Grohe2017}.

\begin{defi}
A \emph{tree decomposition} of a $\sigma$-structure $\str{A}$ is a pair $(T,B)$ with $T$ a tree and $B: T \rightarrow 2^A$ such that:
\begin{enumerate}
\item For every $a \in A$ the set $\{ t \ \mid \ a \in B(t) \}$ induces a subtree of
$T$; and
\item for all relational symbols $R \in T$ and related tuples $\tup{a}\in R^{\str{A}}$, there
exists a node $t \in T$ such that  $\tup{a} \subset B(t)$.
\end{enumerate}
\end{defi}

To arrive at a generalisation of tree decomposition which allows for the non-locality discussed above, we first introduce the following extension of ordinary tree decompositions.

\begin{defi}
An \emph{extended tree decomposition} of a $\sigma$-structure $\str{A}$ is a triple $(T,\beta,\gamma)$ with $\beta,\gamma : T \rightarrow 2^A$ such that:
\begin{enumerate}
\item $(T,B)$ is a tree-decomposition of $\str{A}$ where $B : T \rightarrow 2^A$ is defined by $B(t) \defeq \beta(t) \cup \gamma(t)$; and
\item if $a \in \gamma(t)$ and $a \in B(t')$ then $t \leq t'$.
\end{enumerate}
\end{defi}

In an extended tree decomposition, the \emph{bags} $B$ of the underlying tree decomposition are split into a \emph{fixed bag} $\beta$ and a \emph{floating bag} $\gamma$. The second condition above ensures that $\gamma(t)$ contains only elements $a\in \str{A}$ for which $t$ is their \textit{first}\footnote{minimum in the tree order} appearance in $(T,B)$.
Width and arity are two important properties of extended tree decompositions.

\begin{defi}
  Let $D = (T,\beta,\gamma)$ be an extended tree decomposition.

  The \emph{width}, $w(D)$, of $D$ is $\max_{t \in T}|\beta(t)|$.

  The \emph{arity}, $\text{ar}(D)$, of $D$ is the least $n \leq w(D)$ such that:
  \begin{enumerate}
  \item if $t < t'$ then $|\beta(t') \cap \gamma(t)| \leq n$; and
  \item for every tuple $(a_1,\ldots,a_m)$ in every relation $R$ of $\str{A}$, there is a $t \in T$ such that $\{a_1,\ldots,a_m\} \subseteq B(t)$ and $| \{a_1,\ldots,a_m\} \cap \gamma(t)| \leq n$.
\end{enumerate}
\end{defi}
We note that the definition of width here differs from the width of the underlying tree decomposition $(T,B)$. However as we see in Lemma~\ref{lem:decomp_edecomp} having an ordinary tree decomposition of width $k$ is equivalent to having an extended tree decomposition of width $k$ and arity $1$.

We are particularly interested in extended tree decompositions that are further well-structured, in a sense that is related to the definition of structured $n,k$-histories in Section~\ref{sec:comonad}.
\begin{defi}
An extended tree decomposition with width $k$ and arity $n$ is \emph{\good} if for every $a \in A$ there exists $t \in T$ s.t. $a \in \gamma(t)$, for every node $t$, $\gamma(t) \neq \emptyset$, for any child $t^{\prime}$ of $t$ $\beta(t^{\prime})\cap \gamma(t) \neq \emptyset$ and for any $t^{\prime\prime}$ a child of $t^{\prime}$ we have that either:
\begin{itemize}
  \item $|\beta(t^{\prime}) \cap \gamma(t)| = n $; or
  \item $\abs{\beta(t^{\prime})} < k$; or
  \item $\gamma(t)\cap \beta(t^{\prime}) \setminus \beta(t^{\prime\prime}) \neq \emptyset$
\end{itemize}

\end{defi}

\subsection{Drawing extended tree decompositions and examples}
We draw extended tree decompositions as trees where the nodes have two labels, an upper label indicating the fixed bag at that node and the lower label denoting the floating bag.  In this subsection, we give some simple examples of these decompositions.

\begin{exa}
Any structure $\str{A}$ which has no relations of arity greater than $n$ admits a trivial arity $n$, width $0$ extended tree decomposition with a single node. This is drawn as:

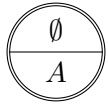
\[
\begin{tikzpicture}
\tikzstyle{every node}=[state with output,draw,double]
\node {
$\emptyset$
\nodepart{lower}
$A$
};
\end{tikzpicture}
\]
\end{exa}

From this example we see that, in particular, any graph $\str{G}$ has a trivial extended tree decomposition of arity $2$. The next two examples show that for graphs, extended tree decompositions of arity $1$ look similar to ordinary tree decompositions.
\begin{exa}
Consider the following tree $\mathcal{T}$ as a graph.
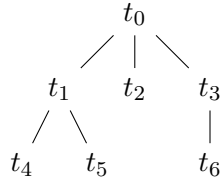
\[
  \begin{minipage}{0.4\textwidth}
        \centering
        \begin{tikzpicture}
        \node (t0) at (0,1) {$t_0$};
        \node (t1) at (-1,0) {$t_1$};
        \node (t2) at (0,0) {$t_2$};
        \node (t3) at (1,0) {$t_3$};
        \node (t4) at (-1.5,-1) {$t_4$};
        \node (t5) at (-0.5,-1) {$t_5$};
        \node (t6) at (1,-1) {$t_6$};
        \draw   (t1) -- (t0) -- (t2)
        (t0) -- (t3) -- (t6)
        (t4) -- (t1) -- (t5);
        \end{tikzpicture}
  \end{minipage}
  \]
As with ordinary tree decompositions a tree can be given a decomposition of width $1$ by creating a bag for each edge. The corresponding extended tree decomposition of width $1$ and arity $1$ for $\mathcal{T}$ is the following:
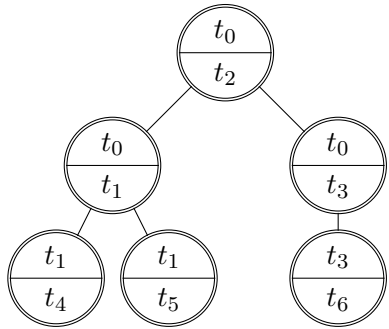
\[
    \begin{minipage}{0.4\textwidth}
        \centering
        \begin{tikzpicture}
        \tikzstyle{every node}=[state with output,draw,double]
        \node (t0) at (0,1.5) {$t_0$
        \nodepart{lower}
        $t_2$
        };
        \node (t1) at (-1.5,0) {$t_0$
        \nodepart{lower}
        $t_1$
        };
        \node (t3) at (1.5,0){$t_0$
        \nodepart{lower}
        $t_3$
        };
        \node (t4) at (-2.25,-1.5) {$t_1$
        \nodepart{lower}
        $t_4$
        };
        \node (t5) at (-0.75,-1.5) {$t_1$
        \nodepart{lower}
        $t_5$
        };
        \node (t6) at (1.5,-1.5) {$t_3$
        \nodepart{lower}
        $t_6$
        };
        \draw   (t1) -- (t0)
        (t0) -- (t3) -- (t6)
        (t4) -- (t1) -- (t5);
        \end{tikzpicture}
\end{minipage}
\]
Unlike with ordinary tree decompositions, the floating bags in extended tree decompositions can be used to give more succinct decompositions (without changing the width). For example, the following is an extended decomposition of $\mathcal{T}$ again with width $1$ and arity $1$.
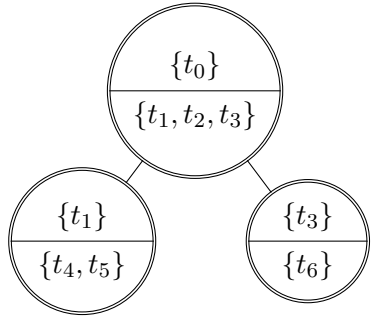
\[
    \begin{minipage}{0.4\textwidth}
        \centering
        \begin{tikzpicture}
        \tikzstyle{every node}=[state with output,draw,double]
        \node (t0) at (0,2) {$\{ t_0 \}$
        \nodepart{lower}
        $\{ t_1, t_2, t_3\}$
        };
        \node (t1) at (-1.5,0) {$\{ t_1 \}$
        \nodepart{lower}
        $\{t_4, t_5 \}$
        };
        \node (t3) at (1.5,0) {$\{ t_3 \}$
        \nodepart{lower}
        $\{t_6 \}$
        };
        \draw   (t1) -- (t0) -- (t3);
        \end{tikzpicture}
\end{minipage}
\]
\end{exa}

As we see in Lemma~\ref{lem:decomp_edecomp}, the correspondence between ordinary tree decompositions and extended tree decompositions of arity $1$ extends beyond trees to all relational structures. However, for signatures of arity higher than $2$ increasing the arity of an extended tree decomposition can result in non-trivial decompositions of lower width as is shown by the following example.

\begin{exa}
  Consider a hypergraph $\mathcal{T}^{\prime}$ constructed from $\mathcal{T}$ above by adding ternary edges $\{t_0, t_1, t_2\}, \{t_0, t_1, t_3\}, \{t_0, t_2, t_3\}$ and $\{t_1, t_4, t_5\}$. Such a structure contains a $4$-clique $\{t_0,t_1, t_2, t_3\}$ in its Gaifman graph (see Libkin~\cite{Libkin2004} Definition 4.1) and so cannot have an ordinary tree decomposition of width less than $3$. However, the following is an extended tree decomposition of width $1$ and arity $2$ for $\mathcal{T}^{\prime}$:

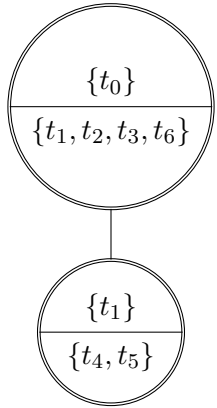
  \[
      \begin{minipage}{0.4\textwidth}
          \centering
          \begin{tikzpicture}
          \tikzstyle{every node}=[state with output,draw,double]
          \node (t0) at (0,2) {$\{ t_0 \}$
          \nodepart{lower}
          $\{ t_1, t_2, t_3, t_6\}$
          };
          \node (t1) at (0,-1) {$\{ t_1 \}$
          \nodepart{lower}
          $\{t_4, t_5 \}$
          };
          \draw   (t1) -- (t0);
          \end{tikzpicture}
  \end{minipage}
  \]

\end{exa}

\subsection{Preliminary results on extended tree decompositions}

Before proving the main result of this section we present two results which establish some basic facts about this new type of decomposition. The first establishes the equivalence of width $k$, arity $1$ extended tree decompositions with ordinary tree decompositions of width $k$. This is interesting as we recall from~\cite{Abramsky2017} that tree decompositions of width $k$ correspond to coalgebras of $\T{k+1}$ whereas we will see in Theorem~\ref{thm:coalgebra} that coalgebras of $\G{1,k}$ give extended tree decompositions of arity $1$ and width $k$. In this light, this result can be seen as demonstrating the extra strength of $\G{1,k}$ over $\T{k}$.

\begin{lem}\label{lem:decomp_edecomp}
  A relational structure $\str{A}$ has a tree decomposition of width $k$ if, and only if, it has an extended tree decomposition of width $k$ and arity 1
\end{lem}

\begin{proof}

($\implies$) Without loss of generality we can assume that $(T, \beta)$ is a tree decomposition such that for all $t \in T$ $|\beta(t)| = k+1$ and if $t^{\prime}$ is a child of $t$ in $T$ then $|\beta(t) \cap \beta(t^{\prime})| = k$. We now show how to transform such a tree decomposition into an extended decomposition $(T^{\prime}, \beta^{\prime}, \gamma)$ of width $k$ and arity $1$.

Define the equivalence relation $\approx$ on $T$ as \[t^{\prime} \approx t^{\prime\prime} \iff \text{$t^{\prime}$ and $t^{\prime\prime}$ have the same parent $t$ in $T$} \textbf{ and } \beta(t) \cap \beta(t^{\prime}) = \beta(t) \cap \beta(t^{\prime\prime})\]

Now we can define the extended decomposition as follows:
\begin{itemize}
  \item $T^{\prime} = T/{\approx}$
  \item $\beta^{\prime}([t]) = \beta(t) \cap \beta(t_0)$ where $t_0$ is the common parent of the elements of $[t]$
  \item $\gamma([t]) = \bigcup_{t^{\prime} \in [t]} \beta(t^{\prime})\setminus \beta(t_0)$
\end{itemize}

For non-root nodes $t$ in $T$ both $\beta^{\prime}$ and $\gamma$ are well-defined by the definition of $\approx$.  For the singleton equivalence class $[r]$ containing the root of $T$ we choose any $c_r \in \beta(r)$ and define $\beta^{\prime}([r]) = \beta(r) \setminus \{c_r\}$ and $\gamma([r]) = \{c_r\}$.

Letting $B([t]) = \beta^{\prime}([t]) \cup \gamma([t])$ we have that  $B([t]) \supset \beta(t)$ and so $(T^{\prime}, B)$ is a tree decomposition.  Furthermore, $\gamma([t]) \cap \beta(t_0) = \emptyset$ by definition, so for any $[t^{\prime}] < [t]$ we have $B([t^{\prime}]) \cap \gamma([t]) = \emptyset$ by the condition that $\beta^{-1}(x)$ is a connected subtree of $T$ for any $x \in T$. So $(T^{\prime}, \beta^{\prime}, \gamma)$ is an extended tree decomposition.

It is easy to see that the maximum size of $\beta^{\prime}(t)$ is equal to $k$ by design. So the width of $(T^{\prime}, \beta^{\prime}, \gamma)$ is $k$. If $\tup{a}$ is a tuple in a relation of $\str{A}$ we know that there is a node $t \in T$ such that $\tup{a} \subset \beta(t)$. By definition, $\beta(t) \subset B^{\prime}(t)$ with $|\beta(t) \cap \gamma([t])| \leq 1$. So the arity of $(T^{\prime}, \beta^{\prime}, \gamma)$ is 1, as required.

  ($\impliedby$)
  To go backwards we take a width $k$, arity $1$ extended tree decomposition $(T, \beta, \gamma)$ and we construct a tree decomposition $(\tilde{T}, \tilde{\beta})$ by replacing each node $t \in T$ with the following spider $H_t$:
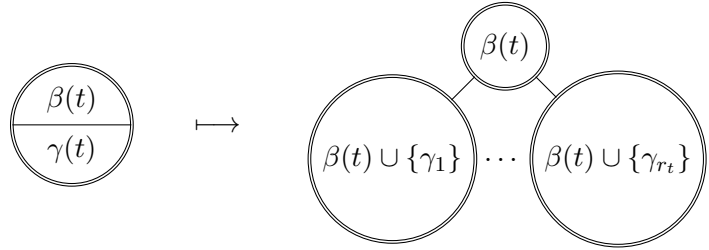
  \[
  \begin{minipage}{0.2\textwidth}
      \centering
  \begin{tikzpicture}
  \tikzstyle{every node}=[state with output,draw,double]
  \node {
  $\beta(t)$
  \nodepart{lower}
  $\gamma(t)$
  };
  \end{tikzpicture}
  \end{minipage}
  \longmapsto
  \begin{minipage}{0.45\textwidth}
      \centering
  \begin{tikzpicture}
  \tikzstyle{every node}=[state,draw,double]
  \node (t0) at (0,0){
  $\beta(t)$
  };
  \node (t1) at (-1.5,-1.5){
  $\beta(t) \cup \{\gamma_1\}$
  };
  \node (t2) at (1.5,-1.5){
  $\beta(t)\cup\{\gamma_{r_t}\}$
  };
  \draw   (t1) -- (t0) -- (t2);
  \path (t1) -- node[draw=none]{\ldots} (t2);
  \end{tikzpicture}
  \end{minipage}
  \]
where the children of the leaf of $H_t$ labelled by $\beta(t) \cup \{\gamma_i\}$ are the roots of the spiders $H_{t^{\prime}}$ such that $t^{\prime}$ is a child of $t$ in $T$ and $\beta(t^{\prime}) \cap \gamma(t) = \{\gamma_i\}$. To see that this is a tree decomposition note firstly that $\tilde{T}$ is clearly a tree under this construction. Next, it is easy to see that for any $a \in \beta(t) \cup \gamma(t)$, $a$ either appears in every bag of $H_t$ or just in a single leaf. This means that the bags containing $a$ in $\tilde{T}$ still form a connected subtree. Lastly, we need to show that each related tuple $\tup{a}$ in $\str{A}$ is contained in some bag of $\tilde{T}$. This is guaranteed by the condition that $(T,\beta, \gamma)$ has arity $1$, which means any time $\tup{a}\subset \beta(t) \cup \gamma(t)$ there exists $\gamma_i \in \gamma(t)$ such that $\tup{a} \subset \beta(t) \cup \gamma(t) \cup \{ \gamma_i \}$.
\end{proof}

Having established the connection between extended tree decompositions and ordinary tree decompositions we now relate extended tree decompositions to our construction in Section~\ref{sec:comonad} with the next easy but important result. It is noteworthy here that the extended tree decompositions admitted by the structures from Section~\ref{sec:comonad} are \emph{structured}. This is important later in this section.

\begin{lem}\label{lem:gnk_decomp}
  For any finite $\str{A}$, there is a structured extended tree decomposition of $\G{n,k}\str{A}$ of width $k$ and arity $n$.
\end{lem}

\begin{proof}
  Recall that the underlying set of $\G{n,k}\str{A}$ consists of representatives $[s | a]$ of equivalence classes in $\T{k}\str{A}/\approx_n$ where $s \in ((A \times [k])^{\leq n})^{\ast}$ is a \emph{structured} $n,k$-history and $a \in A$. We construct an extended tree decomposition where each node is an $n,k$-history $s$ appearing in one of these representatives. The tree ordering is simply given by the prefix relation. The fixed bag at $s$, $\beta(s)$, contains up to $k$ elements which represent the at most $k$ elements which are pebbled after $s$ is played. To describe these explicitly, let $\overline{s} \in \T{k}\str{A}$ be the flattening of the list $s$ and for each $i\in [k]$ appearing as a pebble index in $s$ and let $s_i$ be the maximal prefix of $\overline{s}$ which ends in $(a,i)$ for some $a\in A$. Then $\beta(s)$ contains the $\approx_n$-equivalence classes of each of the $s_i$. As there can be at most $k$ elements in this set, our extended tree decomposition has width $k$. The floating bag is given, more simply as $\gamma(s) = \{[s|a] \mid a\in A\}$. From this description it is easy to see that for any $[s|a] \in \G{n,k}\str{A}$, if $[s|a]$ appears in $\beta(s^{\prime})$ then $s$ is a prefix of $s^{\prime}$ and for any $s^{\prime\prime}$ with $s\sqsubset s^{\prime\prime} \sqsubset s^{\prime}$ we have $[s | a] \in \beta(s^{\prime\prime})$. This confirms that $B^{-1}([s|a])$ is a connected subtree of T and that $\gamma^{-1}([s|a])$ is a singleton containing the root of that subtree.

  To show that $(T,\beta, \gamma)$ defines an extended tree
  decomposition of $\G{n,k}\str{A}$ it now suffices to show that any
  related tuple $\tup{g} = ([s_1|a_1], \dots [s_l|a_l]) \in
  R^{\G{n,k}\str{A}}$ appears in some bag. Because of the way
  relations are defined in $\G{n,k}$ we can find $(t_1, \dots t_l) \in
  R^{\T{k}\str{A}}$ s.t. $q(t_i)= [s_i|a_i]$. By the definition of
  relations in $\T{k}\str{A}$ we know that the $t_i$ are totally
  ordered by the prefix relation. This means that the $s_i$ is
  similarly  totally ordered with largest element $s$. The related
  tuple is contained in $\beta(s)\cup\gamma(s)$. Furthermore, $\tup{g}\cap \gamma(s)$ contains the $t_i$ for which $q(t_i)= [s|a_i]$. As these are linearly ordered by the prefix relation it would be impossible for there to be more than $n$ distinct such lists. This means that ($T,\beta, \gamma$) is indeed an extended tree decomposition of width $k$ and arity $n$. 

  To see that $(T, \beta, \gamma)$ is structured we rely on the fact that the sequences $s\in ((A\times[k])^{\leq n})^{\ast}$ appearing in $T$ are themselves structured in the sense of Definition~\ref{def:nk_structured}. The proof is as follows. Suppose there is a node $s \in T$ with a child $s;x \in T$ where $x \in (A\times[k])^{\leq n}$ and suppose that  $|\beta(s;x) \cap \gamma(s)| < n$ and $\beta(s;x) = k$. We now need to show that for any node $s;x;y \in T$ $\gamma(s)\cap\beta(s;x)\setminus\beta(s;x;y) \neq \emptyset$. Unpacking the definitions we have that $\gamma(s)\cap\beta(s;x)$ contains elements $[s | a]$ where $(a, i)$ appears in $x$ for some $i$. As we also know that $|\beta(s;x)| = k$, which means in particular that $x$ does not contain two pairs $(a,i)$ $(a,j)$ for $i \neq j$ because if it did the contributions from pebbles $i$ and $j$ to $\beta(s;x)$ would both be $[s | a]$. These two facts together mean that the length of $x$ must be strictly less than $n$. Thus as $s;x;y$ is a structured $n,k$-history we must have that the first element of $y$ is $(b_y,i_y)$ where such the index $i_y$ appears in some pair $(b_x,i_y)$ in $x$. It is not hard to see that $[s|b_x] \in \gamma(s)\cap\beta(s;x)\setminus\beta(s;x;y)$, completing our proof.
\end{proof}

We now prove the main claim of this section: the $\G{n,k}$-coalgebras are in correspondence with structured extended tree decompositions of width $k$ and arity $n$.
\subsection{\texorpdfstring{Correspondence with $\G{n,k}$ coalgebras}{Correspondence with H\textsubscript{n,k}-coalgebras}}

In this final subsection we establish the connection between width $k$, arity $n$ extended decompositions of a $\str{A}$ which are structured with coalgebras $\alpha: \str{A} \rightarrow \G{n,k}\str{A}$. Formally stated, we prove the following theorem:

\begin{thm}\label{thm:coalgebra}
  For $\str{A}$ a finite relational structure the following are equivalent:
  \begin{enumerate}
    \item there is a $\G{n,k}$-coalgebra $\alpha: \str{A} \rightarrow \G{n,k}\str{A}$
    \item there is a \good extended tree decomposition of $\str{A}$ with width at most $k$ and arity at most $n$
  \end{enumerate}
\end{thm}

\begin{proof}
  \textit{(1 $\implies$ 2)} Let $\alpha$ be a coalgebra and, as $\epsilon\circ\alpha = \text{id}_\str{A}$, let $\alpha(a) = [s_a |a]$. Recall that by Lemma~\ref{lem:gnk_decomp} there is a \good extended tree decomposition $(T,\beta, \gamma)$ of $\G{n,k}\str{A}$ with arity $n$ and width $k$ where the nodes of $T$ are labelled by structured $n,k$-histories $s \in ((A\times [k])^{\leq n})^{\ast}$. We use this decomposition to define a decomposition $(T_\alpha, \beta_{\alpha}, \gamma_{\alpha})$ on $\str{A}$ as follows:
  \begin{itemize}

    \item $T_\alpha$ is the tree $T$ restricted to the set $\{s_a \ | \ a \in A\}$.

    \item $\beta_\alpha(s) \defeq \{a \in A \ | \ \alpha(a) \in \beta(s) \}$.

    \item $\gamma_\alpha(s) \defeq \{a \in A \ | \ \alpha(a) \in \gamma(s) \}$.

  \end{itemize}
  We now show, firstly, that this is an extended tree decomposition, secondly that it has width $k$ and arity $n$ and finally that it is structured.

  \paragraph{$(T_\alpha, \beta_\alpha, \gamma_\alpha)$ is an extended tree decomposition}
First of all this requires that $T_\alpha$ be a tree. For any $s \in T_\alpha$ we have some $a \in A$ with $\alpha(a) = [s | a]$. Suppose that $s = [l_1|l_2|\dots| l_m]$.  It is sufficient to show that $s_i \in T_\alpha$ for any prefix $s_i = [l_1|\dots|l_i]$ of $s$ (including the empty sequence). This fact can be deduced from the comultiplication law that for all $a$ $\G{n,k}\alpha( \alpha(a)) = \delta_A(\alpha(a))$. The left-hand side of this equation is  $\G{n,k}\alpha(\alpha(a)) = [\overline{s} | \alpha(a)]$ where $\overline{s} = [\overline{l_1}|\overline{l_2}|\dots| \overline{l_m}]$ and the right-hand side is $\delta_A(\alpha(a)) = [\tilde{s} | \alpha(a)]$  where $\tilde{s} = [\tilde{l_1}|\tilde{l_2}|\dots| \tilde{l_m}]$. Taking any $l_i = [(b_1, p_1), \dots (b_{m_i},p_{m_i})]$ it is not hard to see that $\overline{l_i} = [([\alpha(b_1)|b_1],p_1) \dots ([\alpha(b_{m_i})|b_{m_i}],p_{m_i})]$ and $\tilde{l_i} = [([s_{i-1} | b_1],p_1), \dots ([s_{i-1} | b_{m_i}],p_{m_i})]$. From this we can conclude that for any $b$ appearing in $l_i$ for any $1\leq i \leq m$ we have that $\alpha(b) = [s_{i-1} | b]$ where $s_0$ is the empty sequence. This proves that all prefixes of $s$ appear in $T_\alpha$. Now we show that with $B_\alpha \defeq \beta_\alpha \cup \gamma_\alpha$ $(T_\alpha, B_\alpha)$ defines a tree decomposition of $\str{A}$. Indeed $B_\alpha^{-1}(a)$ is a subtree because it is really the intersection of two subtrees of the original $T$. Furthermore, for any $\tup{a} \in R^{\str{A}}$, we have that $\alpha(\tup{a})\in R^{\G{n,k}\str{A}}$. As $(T,\beta, \gamma)$ is a tree decomposition, there is an $s\in T$ with $\alpha(\tup{a}) \subset \beta(s) \cup \gamma(s)$. You can assume $\alpha(\tup{a}) \cap \gamma(s) \neq \emptyset$ by taking the longest prefix of $s$ which satisfies this\footnote{This works by noting that for $s'$ a parent of $s$ in $T$, $\beta(s)\setminus\beta(s') \subset \gamma(s)$}. This means that $s \in T_\alpha$ and $\tup{a} \subset \beta_\alpha(s) \cup \gamma_\alpha(s)$. This shows that $(T_\alpha, \beta_\alpha \gamma_\alpha)$ defines an extended tree decomposition.

\paragraph{$(T_\alpha, \beta_\alpha \gamma_\alpha)$ has width $k$ and arity $n$}
As $\alpha$ is injective by the coalgebra law $\epsilon\circ\alpha = \text{id}_{\str{A}}$, we know that for any $s \in T_\alpha$ $|\beta_\alpha(s)| \leq |\beta(s)|$  and $|\beta_\alpha(s)| \leq |\beta(s)|$ by definition. As $(T,\beta, \gamma)$ has width $k$ this means that $|\beta_\alpha(s)| \leq k$ for all $s \in T_\alpha$ and so $(T_\alpha, \beta_\alpha, \gamma_\alpha)$ has width $k$. For arity, we have that  for any related tuple $\tup{a}$ in $\str{A}$ the tuple $\alpha(\tup{a})$ is related in $\G{n,k}\str{A}$. $(T,\beta, \gamma)$ having arity $n$ means that for any $s\in T$ $|\alpha(\tup{a})\cap \gamma(s)| \leq n$. So again by the injectivity of $\alpha$ $|\tup{a}\cap \gamma_\alpha(s)| \leq n$ and so $(T_\alpha, \beta_\alpha, \gamma_\alpha)$ has arity $n$.

\paragraph{$(T_\alpha, \beta_\alpha, \gamma_\alpha)$ is \good} Finally the extended tree decomposition is \good because $(T,\beta, \gamma)$ is \good and the coalgebra laws guarantee that $|\beta_\alpha(s)| = |\beta(s)|$ and $|\beta_\alpha(s')\cap \gamma_\alpha(s)| = |\beta(s')\cap \gamma(s)|$  for any $s \in T_\alpha$ with child node $s'$. This first equation is deduced by noting that injectivity guarantees $|\beta_\alpha(s)| \leq |\beta(s)|$. The reverse inequality comes from the fact that any $t \in \beta(s)$ is the $\approx_{n}$ equivalence class of some prefix of $s$. As we saw before, the comultiplication law guarantees that such classes are realised as $\alpha(b)$ for an appropriate $b$ so we have $|\beta_\alpha(s)| = |\beta(s)|$. The second equation follows from the same reasoning. Together these ensure that the conditions for being \good which are satisfied in $(T, \beta, \gamma)$ are also satisfied in $(T_\alpha, \beta_\alpha, \gamma_\alpha)$.

  \textit{(2 $\implies$ 1)} Defining a $\G{n,k}$ coalgebra from a
  \good extended tree decomposition $(T, \beta, \gamma)$ of width $k$
  and arity $n$ requires some careful bookkeeping which is presented
  explicitly here. Throughout we rely on the fact that our tree $T$
  comes with an order $\leq$ and so has a root which we call $r$. By the conditions of being \good, we have for each $a \in A$ a $\leq$-minimal node $c_a \in A$ where $a$ appears in $B(c_a)$ and we have that $a\in \gamma(c_a)$. This means in particular that at the root $\beta(r) = \emptyset$.

  The general strategy in defining the coalgebra $\alpha_T$ is to assign to each node $c \in T$ a structured $n,k$-history $s_c \in ((A\times[k])^{\leq n})^{\ast}$ which records the elements of $A$ which have appeared in $(T,\beta,\gamma)$ on the path from $r$ to $c$. We then show that $\alpha_T(a) = [s_{c_a} | a]$ defines a $\G{n,k}$-coalgebra for $\str{A}$.

\paragraph{Defining $s_c$}
Starting at the root we define $s_r$ to be the empty list. At each new node in $c \in T$ with parent $c'$ we define $l_c \in (A\times [k])^{\leq n}$ to record the elements of $A$ which appear in $\gamma(c')$ and persist in $\beta(c)$. As the arity of $(T,\beta,\gamma)$ is $n$ we know that $|\gamma(c')\cap\beta(c)|\leq n$.  We then form $s_c$ by appending $l_c$ to $s_{c'}$. This inductively defines $s_c$ on all the nodes of $T$.

Defining $l_c$ in such a way as to ensure $s_c$ is a structured $n,k$-history requires some care with assigning pebble indices from $[k]$ to the elements in $\gamma(c') \cap \beta(c)$. To help keep track of these indices we also define a function $\iota_c : \beta(c) \rightarrow [k]$. We say that a \emph{live} prefix of $s_c$ is a prefix $s'$ of the flattened list $F(s_c) \in (A\times [k])^{\ast}$ with final element $(b, i)$ such that no larger prefix of $F(s_c)$ ends with $(b', i)$ for any $b' \in A$. We say that $b$ is live in $s_c$ if it appears at the end of some live prefix $s'$. The end goal is that $s_c$ will be an $n,k$-history where the live elements are exactly those in $\beta(c)$ and that for each such element $b$ there is a live prefix of $s_c$ ending in the pair $(b,\iota_c(b))$.

At each $c$ we partition $\beta(c)$ as $N_c \cup R_c$ where $N_c \defeq \gamma(c') \cap \beta(c)$ is the set of \textit{new} elements in $\beta_c$ and $R_c \defeq \beta(c) \cap \beta(c')$ is the set of elements \textit{retained} from the parent node. Firstly, we define $\iota_c$ to be equal to $\iota_{c'}$ on $R_c$. As the width of $(T, \beta, \gamma)$ is $k$ we know that $|\beta(c)| \leq k$ and so the number of \textit{free indices} $|[k] \setminus \iota_c(R_c)|$ is at least as big as the number of new elements $|N_c|$ so we can assign to each element $b$ of $N_c$ a distinct index $\iota_c(b)$ from $[k] \setminus \iota_c(R_c)$. In many cases this is enough and we can pick any ordering $b_1, \dots b_m$ of the elements in $N_c$ and set $l_c$ to be the list $[(b_1,\iota_c(b_1)),\dots (b_m, \iota_c(b_m))]$.

We now need to define some modifications to this to ensure that $s_c$ is structured. Recall that an $n,k$-history $s$ is structured if and only if for every pair of successive blocks $l'$ appearing immediately before $l$ in $s$ we have that either $|l| = n$ or the first pebble index in $l$ must have appeared in $l'$. To ensure this holds true for each $s_c$, we need to take extra care defining $l_c$ in cases where $|l_{c'}|$ or $|N_c|$ are less than $n$.

If $|l_{c'}| < n$ then we must choose $\iota_c(b_1)$ to be an index which appeared in $l_{c'}$. To see that we can do this recall that $(T,\beta, \gamma)$ is structured and so for each non-root node $c'$ with child $c$ we have (using our new language from this proof) that at least one of the following is true
\begin{enumerate}
\item $|N_{c'}| = n$,
\item $|\beta(c')| < k$; or
\item $R_{c}\setminus N_{c'} \neq \emptyset$.
\end{enumerate}
In the first case, we have $|l_{c'}| = n$ so no action needs to be taken.
In the second case, where $|N_{c'}| < n$ and $|\beta(c')| < k$ then there is a spare index $i \in [k]\setminus\iota_{c'}(\beta(c'))$ and we define $l_{c'}$ to be $[(b'_1,\iota_{c'}(b'_1)),\dots (b'_m, \iota_{c'}(b'_{m'}), (b'_{m'}, i)]$ and we define $\iota_c(b_1) \defeq i$.
In the third case, there may not be a spare index $i$ but instead there is some element $b \in N_{c'}\setminus R_{c}$ meaning that some element which appears in $l_{c'}$ does not need to be live after $l_c$. In this case we simply define $\iota_c(b_1) \defeq \iota_{c'}(b)$.
Collectively, these modifications ensure that $s_c$ is structured and so the definition $\alpha_T(a) \defeq [s_{c_a} | a]$ is well-defined. It remains to show that $\alpha_T$ is a coalgebra.

\paragraph{$\alpha_T$ is a coalgebra}
To show that $\alpha_T$ is a homomorphism, take any related tuple $\tup{a} \in R^{\str{A}}$. As $(T,\beta, \gamma)$ is an extended tree decomposition there is some $c$ such that $\tup{a} \subset \beta(c) \cup \gamma(c)$. Now as the arity of the decomposition is $n$ there are at most $n$ elements $a \in \tup{a}$ with $a \in \gamma(c)$ and so $\alpha(a) = [s_c | a]$. For all the other elements $a' \in \tup{a}$ there must be some earlier $c_0$ with $a' \in \gamma(c_0)$ and a unique path $c_0 < c_1 < \dots <c_q = c$ linking $c_0$ and $c$ in $T$. We must have $a' \in \beta(c_1)$ and $a' \in R_{c_i}$ for all $1\leq i \leq q$ so by the definition of $s_c$ above we know that the index $\iota_{c_1}(a')$ used to pebble $a'$ in $l_{c_1}$ has not been reallocated by the end of $s_c$. From this it is easy to see that the tuple $\alpha(\tup{a})$ (with function application defined component-wise on the tuple) is related in $\G{n,k}\str{A}$.
Finally, we verify that $\alpha$ satisfies the coalgebra laws. The counit law, $\epsilon\circ\alpha = \text{id}_{\str{A}}$ is satisfied by definition. For comultiplication, it suffices to check that for any $a,b \in A$, if $b$ appears in $s_{c_a} = [l_{c_1}|\dots|l_{c_q}]$ then it appears in exactly one of the $l_{c_i}$ and $\alpha(b) = [[l_{c_1}|\dots|l_{c_{i-1}}] | b]$. This can be seen to hold from the construction above, concluding our proof.
\end{proof}

\section{Concluding Remarks}\label{sec:conclusion}

The work of Abramsky et al., giving comonadic accounts of pebble games and their relationship to logic has opened up a number of avenues of research.  It raises the possibility of studying logical resources through a categorical lens and introduces the notion of \emph{coresources}.  This view has been applied to pebble games~\cite{Abramsky2017}, Ehrenfeucht-Fra\"iss\'e games, bisimulation games~\cite{Abramsky2018} and also to quantum resources~\cite{Abramsky2017q,AbramskyBKM19}.  In this paper we have extended this approach to logics with generalised quantifiers.

The construction of the comonad $\G{n,k}$ introduces interesting new techniques to this project.  The pebbling comonad $\T{k}$ is graded by the value of $k$ which we think of as a \emph{coresource} increasing which constrains the morphisms.  The new parameter $n$ provides a second coresource, increasing which further constrains the moves of Duplicator.  It is interesting that the resulting comonad can be obtained as a quotient of $\T{k}$ and the strategy lifting argument developed
in Section~\ref{sec:comonad} could prove useful in other contexts.
The morphisms in the Kleisli category correspond to winning strategies in a new game we introduce which characterises a natural logic: the positive logic of homomorphism-closed quantifiers.  The isomorphisms correspond to an already established game: Hella's $n$-bijective game with $k$ pebbles.  This relationship allows for a systematic exploration of variations characterising a number of natural fragments of the logic with $n$-ary quantifiers.  One natural fragment that is not yet within this framework and worth investigating is the logic of embedding-closed quantifiers of Haigora and Luosto~\cite{Haigora2014}.

This work opens up a number of perspectives.  Logics with generalised quantifiers have been widely studied in finite model theory.  They are less of interest in themselves and more as tools for proving inexpressibility in specific extensions of first-order or fixed-point logic.  For instance, the logics with rank operators~\cite{Dawar2009,Gradel2019} of great interest in descriptive complexity  have been analysed as fragments of a more general logic with linear-algebraic quantifiers~\cite{Dawar2019}.  It would be interesting to explore whether the comonad $\G{n,k}$ could be combined with a vector space construction to obtain a categorical account of this logic.

More generally, the methods illustrated by our work could provide a way to deconstruct pebble games into their component parts and find ways of constructing entirely new forms of games and corresponding logics.  The games we consider and classify are based on Duplicator playing different kinds of functions (i.e.\ morphisms on finite sets) and maintaining different kinds of homomorphisms (i.e.\ morphisms in the category of $\sigma$-structures).  Could we build reasonable pebble games and logics on other categories?  In particular, can we bring the algebraic pebble games of~\cite{Dawar2017} into this framework?

\bibliographystyle{alphaurl}
\bibliography{lmcs_submission}
\end{document}